\documentclass[journal]{IEEEtran}
\usepackage{tipa}
\usepackage{amsfonts}
\usepackage{amsfonts}
\usepackage{amsfonts}
\usepackage{mathbbold}
\usepackage{amsfonts}
\usepackage{amssymb}
\usepackage{stfloats}
\usepackage{cite}
\usepackage{graphicx}
\usepackage{psfrag}
\usepackage{subfigure}
\usepackage{amsmath}
\usepackage{array}

\DeclareMathOperator*{\argmin}{\arg\min}

\begin{document}

\title{Delay-Optimal User Scheduling and Inter-Cell
Interference Management in Cellular Network via Distributive
Stochastic Learning}

\newtheorem{Thm}{Theorem}
\newtheorem{Lem}{Lemma}
\newtheorem{Cor}{Corollary}
\newtheorem{Def}{Definition}
\newtheorem{Exam}{Example}
\newtheorem{Alg}{Algorithm}
\newtheorem{Prob}{Problem}
\newtheorem{Rem}{Remark}
\newtheorem{Proof}{Proof}
\newtheorem{Subproblem}{Subproblem}
\newtheorem{assumption}{Assumption}

\author{\authorblockN{Huang Huang {\em Student Member, IEEE}, Vincent K. N. Lau, {\em Senior Member, IEEE}}
\thanks{ The authors are with
the Department of Electronic and Computer Engineering (ECE), The
Hong Kong University of Science and Technology (HKUST), Hong Kong.
(email: huang@ust.hk, eeknlau@ee.ust.hk).}}

\markboth{To be appeared in IEEE Trans. Wireless Commun.}
{Shell
\MakeLowercase{\textit{et al.}}: Bare Demo of IEEEtran.cls for
Journals}

\maketitle

\begin{abstract}
In this paper, we propose a distributive queue-aware intra-cell user
scheduling and inter-cell interference (ICI) management control
design for a delay-optimal celluar downlink system with $M$ base
stations (BSs), and $K$ users in each cell. Each BS has $K$ downlink
queues for $K$ users respectively with heterogeneous arrivals and
delay requirements. The ICI management control is adaptive to joint
queue state information (QSI) over a slow time scale, while the user
scheduling control is adaptive to both the joint QSI and the joint
channel state information (CSI) over a faster time scale. We show
that the problem can be modeled as an infinite horizon average cost
Partially Observed Markov Decision Problem (POMDP), which is NP-hard
in general. By exploiting the special structure of the problem, we
shall derive an equivalent Bellman equation to solve the POMDP
problem. To address the distributive requirement and the issue of
dimensionality and computation complexity, we derive a distributive
online stochastic learning algorithm, which only requires local QSI
and local CSI at each of the $M$ BSs. We show that the proposed
learning algorithm converges almost-surely (with probability 1) and
has significant gain compared with various baselines.  The proposed
solution only has linear complexity order $O(MK)$.

\end{abstract}


\section{Introduction}\label{sec:intro}

It is well-known that cellular systems are {\em interference
limited} and there are a lot of works to handle the {\em inter-cell
interference} (ICI) in cellular systems.  Specifically, the optimal
binary power control (BPC) for the sum rate maximization has been
studied in \cite{BPC:2008}. They showed that BPC could provide
reasonable performance compared with the multi-level power control
in the multi-link system. In \cite{pattern:2009}, the authors
studied a joint adaptive multi-pattern reuse and intra-cell user
scheduling scheme, to maximize the long-term network-wide utility.
The ICI management runs at a slower scale than the user selection
strategy to reduce the communication overhead. In
\cite{multicell:cooperation:2008} and the reference therein,
cooperation or coordination is also shown to be a useful tool to
manage ICI and improve the performance of the celluar network.

However, all of these works have assumed that there are infinite
backlogs at the transmitter, and the control policy is only a
function of channel state information (CSI). In practice,
applications are delay sensitive, and it is critical to optimize the
delay performance in the cellular network. A systematic approach in
dealing with delay-optimal resource control in general delay regime
is via Markov Decision Process (MDP) technique. In
\cite{Delay_IT:2006,Vincent:MIMO}, the authors applied it to obtain
the delay-optimal cross-layer control policy for broadcast channel
and point-to-point link respectively. However, there are very
limited works that studied the delay optimal control problem in the
cellular network. Most existing works simply proposed heuristic
control schemes with partial consideration of the queuing
delay\cite{multicell:multiuser:2009}. As we shall illustrate, there
are various technical challenges involved regarding delay-optimal
cellular network.

\begin{itemize}
\item{\bf Curse of Dimensionality:} Although MDP
technique is the systematic approach to solve the delay-optimal
control problem, a primal difficulty is the curse of
dimensionality\cite{Bertsekas:2007}. For example, a huge state space
(exponential in the number of users and number of cells) will be
involved in the MDP and brute force value or policy iterations
cannot lead to any implementable solution\footnote{For a celluar
system with 5 BSs, 5 users served by each BS, a buffer size of 5 per
user and 5 CSI states for each link between one user and one BS, the
system state space contains
$(5+1)^{5\times5}\times5^{5\times5\times5}$ states, which is already
unmanageable.} \cite{RL:survey,Powell:2007}. Furthermore, brute
force solutions require explicit knowledge of transition probability
of system states, which is difficult to obtain in the complex
systems.

\item{\bf Complexity of the Interference Management:} Jointly optimal ICI
management and user scheduling requires heavy computation overhead
even for the throughput optimization problem \cite{pattern:2009}.
Although grouping clusters of cells \cite{BPC:2008} and considering
only neighboring BSs \cite{BSs:neighbour} were proposed to reduce
the complexity, complex operations on a slot by slot basis are still
required, which is not suitable for the practical implementation.

\item{\bf Decentralized Solution:} For delay-optimal multi-cell control,
the entire system state is characterized by the global CSI (CSI from
any BS to any MS) and the global QSI (queue length of all users).
Such system state information are distributed locally at each BS and
centralized solution (which requires global knowledge of the CSI and
QSI) will induce substantial signaling overhead between the BSs and
the Base Station Controller (BSC).
\end{itemize}

In this paper, we consider the delay-optimal inter-cell ICI
management control and intra-cell user scheduling for the cellular
system. For implementation consideration, the ICI management control
is computed at the BSC at a longer time scale and it is adaptive to
the QSI only. On the other hand, the intra-cell user scheduling
control is computed distributively at the BS at a smaller time scale
and hence, it is adaptive to both the CSI and QSI. Due to the {\em
two time-scale} control structure, the delay optimal control is
formulated as an infinite-horizon average cost Partially Observed
Markov Decision Process (POMDP). 
Exploiting the special structure, we propose an {\em equivalent
Bellman Equation} to solve the POMDP. Based on the equivalent
Bellman equation, we propose a distributive online learning
algorithm to estimate a per-user value function as well as a
per-user $\mathbb{Q}$-factor\footnote{The $\mathbb{Q}$-factor
$\mathbb{Q}(s,a)$ is a function of the system state $s$ and the
control action $a$, which represents the {\em potential cost} of
applying a control action $a$ at the current state $s$ and applying
the action $a^{\prime}=\arg\min_{a}\mathbb{Q}(s^{\prime},a)$ for any
system state $s^{\prime}$ in the future\cite{Cao:2007}.}. Only the
local CSI and QSI information is required in the learning process at
each BS. We also establish the technical proof for the almost-sure
convergence of the proposed distributive learning algorithm. The
proposed algorithm is quite different from the iterative update
algorithm for solving the deterministic NUM \cite{Palomar:NUM:2006},
where the CSI is always assumed to be quasi-static during the
iterative updates. However, the delay-optimal problem we considered
is stochastic in nature, and during the iterative updates, the
system state will not be quasi-static anymore. In addition, the
proposed learning algorithm is also quite different from
conventional stochastic learning\cite{Cao:2007,Q-learning:2007}. For
instance, conventional stochastic learning requires centralized
update and global system state knowledge and the convergence proof
follows from standard {\em contraction mapping}
arguments\cite{Bertsekas:2007}. However, due to the distributive
learning requirement and simultaneous learning of the per-user value
function and $\mathbb{Q}$-factor, it is not trivial to establish the
contraction mapping property and the associated convergence proof.
We also illustrate the
performance gain of the proposed solution against various baselines
via numerical simulations. Furthermore, the solution has linear
complexity order $O(MK)$ and it is quite suitable for the practical
implementation.

\section{System Model}\label{sec:model}
In this section, we shall elaborate the system model, as well as the
control policies. We consider the downlink of a wireless celluar
network consisting of $M$ BSs, and there are $K$ mobile users in
each cell served by one BS. Specifically, let
$\mathcal{M}=\{1,...,M\}$ and $\mathcal{K}_m=\{1,...,K\}$ denote the
set of BSs and the set of users served by the BS $m$ respectively.
$k\in\mathcal{K}_m$ denotes the $k$-th user served by BS $m$.
The time dimension is partitioned into {\em scheduling
slots} (every slot lasts for $\tau$ seconds). The system model is
illustrated in Fig.\ref{fig:sys_model}.

\begin{figure}
 \begin{center}
  \resizebox{9cm}{!}{\includegraphics{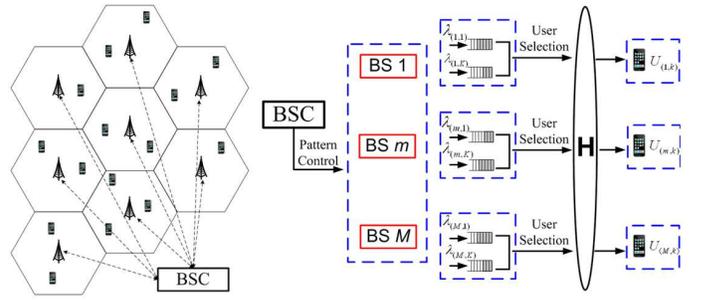}}
 \end{center}
    \caption{Physical layer and queueing model of celluar network.}
    \label{fig:sys_model}
\end{figure}

\subsection{Source Model}
In each BS, there are $K$ independent application streams dedicated
to the $K$ users respectively. Let
$\mathbf{A}(t)=\{\mathbf{A}_m(t)\}_{m=1}^{M}$ and
$\mathbf{A}_m(t)=\{A_{(m,k)}(t)\}_{k=1}^{K}$, where $A_{(m,k)}(t)$
represents the new arrivals (number of bits) for the user
$k\in\mathcal{K}_m$ at the end of the slot $t$.
\begin{assumption}[Assumption on Source
Model]\label{ass:source_model} We assume that the arrival process
$A_{(m,k)}(t)$ is i.i.d over the scheduling slot $t$ according to a
general distribution $\Pr\{A_{(m,k)}\}$ with average arrival rate
$\lambda_{(m,k)}=\mathbb{E}[A_{(m,k)}]$, and the arrival processes
for all the users are independent with each other, i.e.,
$\Pr\{A_{(m,k)}A_{(n,l)}\}=\Pr\{A_{(m,k)}\}\Pr\{A_{(n,l)}\}$ if
$m\neq n$ or $k\neq l$. ~\hfill \IEEEQED
\end{assumption}


Let $\mathbf{Q}(t)=\{\mathbf{Q}_m(t)\}_{m=1}^{M}\in\mathcal{Q}$
denote the global QSI in the system, where $\mathcal{Q}$ is the
state space for the global QSI.
$\mathbf{Q}_m(t)=\{Q_{(m,k)}(t)\}_{k=1}^K$ denotes the QSI in the BS
$m$, where $Q_{(m,k)}(t)$ represents the number of bits for user
$k\in\mathcal{K}_m$ at the beginning of the slot $t$, and $N_{Q}$
denotes the maximal buffer size (bits). When the buffer is full,
i.e, $Q_{(m,k)}=N_Q$, new bits arrivals will be dropped. The
cardinality of the global QSI  is $I_Q=(1+N_Q)^{MK}$.

\subsection{Channel Model and Physical Layer Model}
Let $H_{(m,k)}^n(t)$ and $L_{(m,k)}^n$ denote the small scale
channel fading gain and the path loss from the $n$-th BS to the user
$k\in\mathcal{K}_m$ respectively, and
$\mathbf{H}_{(m,k)}(t)=\{H_{(m,k)}^n(t)\}_{n=1}^M$ is the local CSI
states for user $k$.
$\mathbf{H}_m(t)=\{\mathbf{H}_{(m,k)}(t)\}_{k=1}^K$ denotes the
local CSI states for BS $m$, and the global CSI is denoted as
$\mathbf{H}(t)=\{\mathbf{H}_{m}(t)\}_{m=1}^M\in\mathcal{H}$, where
$\mathcal{H}$ is the state space for the global CSI.
\begin{assumption}[Assumption on Channel Model]\label{ass:csi_model}
We assume that the global $\mathbf{H}$ is quasi-static in each slot.
Furthermore, $H_{(m,k)}^n(t)$ is i.i.d over the scheduling slot $t$
according to a general distribution $\Pr\{H_{(m,k)}^n\}$ and the
small scale channel fading gains for all users are independent with
each other. The path loss $L_{(m,k)}^n$ remains constant for the
duration of the communication session. ~\hfill \IEEEQED
\end{assumption}

The cellular system shares a single common channel with bandwidth
$W$Hz (all the BSs use the same channel). At the beginning of each
slot, the BS is either turned on (with transmit power $P^m_{\max}$)
or off (with transmit power 0)\footnote{Note that the on-off BS
control is shown to be close to optimal
in\cite{BPC:2008,pattern:2009}. Moreover, the solution framework can
be easily extended to deal with discrete BS power control.},
according to a {\em ICI management control policy}, which is defined
later. At each slot, a BS can select only one user for its data
transmission. Specifically, let
$\mathbf{p}=\{p_{m}\}_{m=1}^M\in\mathcal{P}$ denotes an ICI
management control pattern, where $p_{m}=1$ denotes BS $m$ is
active, $p_{m}=0$ otherwise, and $\mathcal{P}$ denotes the set of
all possible control patterns. Furthermore, let
$\mathcal{M}_\mathbf{p}\in\mathcal{M}$ be the set of BSs activated
by the pattern $\mathbf{p}$ and $\mathcal{P}_m\in\mathcal{P}$ be the
set of patterns that activate the BS $m$. The signal received by the
user $k\in\mathcal{K}_m$ at slot $t$, when pattern
$\mathbf{p}\in\mathcal{P}_m$ is selected, is given by
\begin{equation}
\label{eq:sys_model} \begin{array}{l} y_{(m,k)}[t]=
\sqrt{H^m_{(m,k)}L^m_{(m,k)}}x_{m}[t]+\\
\quad\quad\quad\quad\sum\nolimits_{n\neq
m,n\in\mathcal{M}_\mathbf{p}}\sqrt{H^n_{(m,k)}L_{(m,k)}^n}x_{n}[t]+z[t]
\end{array}
\end{equation}
where $x_{m}[t]$ is the transmit signal from the $m$-th BS to the
$k$-th user at slot $t$, and $\{z[t]\}_{t=1}^{\infty}$ is the i.i.d
$\mathcal{N}(0,N_0)$ noise. The achievable data rate of user $k$ can
be expressed by
\begin{equation}
\label{eq:rate}\begin{array}{l} R_{(m,k)}=\\
\left\{\begin{array}{ll}
W\log_2\left(1+\frac{\xi
P^m_{\max}H^m_{(m,k)}L_{(m,k)}^m}{I_{(m,k)}+N_0W}\right)s_{(m,k)}
& \text{if $\mathbf{p}\in\mathcal{P}_m$}\\
0 & \text{otherwise}
\end{array}
\right.
\end{array}
\end{equation}
where $I_{(m,k)}=\sum\limits_{n\neq m,n\in\mathcal{M}_\mathbf{p}}
P^n_{\max}H^n_{(m,k)}L_{(m,k)}^n$, $s_{(m,k)}\in\{0,1\}$ is an
indicator variable with $s_{(m,k)} = 1$ when the user $k$ is
scheduled. $\xi\in(0,1]$ is a constant can be used to model both the
coded and uncoded systems\cite{Vincent:MIMO}.

\subsection{ICI Management and User Scheduling Control Policy}\label{sec:con_pol}
At the beginning of the slot, the BSC will decide which BSs are
allowed to transmit according to a stationary ICI management control
policy defined below.
\begin{Def}[Stationary ICI Management Control Policy]
A stationary ICI management control policy $\Omega_{\mathbf{p}}:
\mathcal{Q}\rightarrow\mathcal{P}$ is defined as the mapping from
current global QSI to an ICI management pattern
$\Omega_{\mathbf{p}}(\mathbf{Q})=\mathbf{p}$. ~\hfill \IEEEQED
\end{Def}



Let $\boldsymbol{\chi}(t)=\{\mathbf{H}(t),\mathbf{Q}(t)\}$ to be the
global system state at the beginning of slot $t$. The active user at
each cell is selected according to a user scheduling policy defined
below.
\begin{Def}[Stationary User Scheduling Policy]
A stationary user scheduling policy
$\Omega_{\mathbf{s}}:\{\mathcal{Q},\mathcal{H}\}\rightarrow\mathcal{S}$
is defined as the mapping from current global system state
$\boldsymbol{\chi}$ to current user scheduling action
$\Omega_{\mathbf{s}}(\boldsymbol{\chi})=\mathbf{s}\in\mathcal{S}$.
The scheduling action $\mathbf{s}$ is a set of all the users'
scheduling indicator variable, i.e., $\mathbf{s}=\{s_{(m,k)},\forall
k\in\mathcal{K}_m,\forall m\}$. It represents which users are
scheduled and which users are not in any given slot. $\mathcal{S}$
is the set of all user scheduling actions. ~\hfill \IEEEQED
\end{Def}

For notation convenience, let
$\Omega=\{\Omega_{\mathbf{p}},\Omega_{\mathbf{s}}\}$ to be the joint
control policy, and
$\Omega(\boldsymbol{\chi})=\{\mathbf{p},\mathbf{s}\}$ be the control
action under state $\boldsymbol{\chi}$.


\section{Problem Formulation}\label{sec:problem}
In this section, we will first elaborate the dynamics of system
state under a control policy $\Omega$. Based on that, we shall
formally formulate the delay-optimal control problem.
\subsection{Dynamics of System State}
Given the new arrival $A_{(m,k)}(t)$ at the end of the slot $t$, the
current system state $\boldsymbol{\chi}(t)$ and the control action
$\Omega(\boldsymbol{\chi}(t))$, The queue evolution for user
$k\in\mathcal{K}_m$ is given by:
\begin{equation}
Q_{(m,k)}(t+1)=\big[\big(Q_{(m,k)}(t)-U_{(m,k)}(t)\big)^++A_{(m,k)}(t)\big]_{\bigwedge
N_Q}
\end{equation}
where $U_{(m,k)}(t)=\lfloor
R_{(m,k)}(\boldsymbol{\chi}(t),\Omega(\boldsymbol{\chi}(t)))\tau\rfloor$
is the number of bits delivered to user $k$ at slot $t$, and
$R_{(m,k)}(\boldsymbol{\chi}(t),\Omega(\boldsymbol{\chi}(t)))$,
given by (\ref{eq:rate}), is the achievable data rate under the
control action $\Omega(\boldsymbol{\chi}(t))$. $\lfloor x\rfloor$
denotes the floor of $x$, $(x)^+=\max(x,0)$, and $(x)_{\bigwedge
N_Q}=\min(x,N_Q)$. Let $\mathbf{U}(t)=\{\mathbf{U}_m(t)\}_{m=1}^M$,
and $\mathbf{U}_m(t)=\{U_{(m,k)}(t)\}_{k=1}^K$, $U_{(m,k)}(t)=
R_{(m,k)}(\boldsymbol{\chi}(t),\Omega(\boldsymbol{\chi}(t)))\tau$
for the user $k\in\mathcal{K}_m$, and
$\mathbf{\hat{Q}}(t+1)=\big[\big(\mathbf{Q}(t)-\mathbf{U}(t)\big)^++\mathbf{A}(t)\big]_{\bigwedge
N_Q}$. Therefore, given a control policy $\Omega$, the random
process $\{\mathbf{H}(t), \mathbf{Q}(t)\}$ is a controlled Markov
chain with transition probability
\begin{equation}
\label{eq:sys_tran}\begin{array}{l}
\Pr\{\boldsymbol{\chi}(t+1)|\boldsymbol{\chi}(t),\Omega(\boldsymbol{\chi}(t))\}=\\
\left\{\begin{array}{ll} \Pr\{\mathbf{H}(t+1)\}\Pr\{\mathbf{A}(t)\}
& \text{if
$\mathbf{Q}(t+1)=\mathbf{\hat{Q}}(t+1)$}\\
0 & \text{otherwise}
\end{array}
\right.
\end{array}
\end{equation}


\subsection{Delay Optimal Control Problem Formulation}
Given a stationary control policy $\Omega$, the average cost of the
user $k\in\mathcal{K}_m$ is given by:
\begin{equation}
\label{eq:T_single}
\overline{T}_{(m,k)}(\Omega)=\lim\sup_{T\rightarrow
\infty}\frac{1}{T}\sum\nolimits_{t=1}^T\mathbb{E}[f(Q_{(m,k)}(t))]
\end{equation}
where $f(Q_{(m,k)})$ is a monotonic increasing cost function of
$Q_{(m,k)}$. For example, when
$f(Q_{(m,k)})=Q_{(m,k)}/\lambda_{(m,k)}$, using Little's Law
\cite{Delay_IT:2006,Ross:2003}, $\overline{T}_{(m,k)}(\Omega)$ is an
approximation\footnote{Strictly speaking, the average delay is given
by $\overline{T}_{(m,k)}(\Omega)=\lim\sup_{T\rightarrow
\infty}\frac{1}{T}\sum\nolimits_{t=1}^T\mathbb{E}[\frac{Q_{(m,k)}}{\lambda_{(m,k)}(1-\text{PBD}_{(m,k)})}]$,
where $\text{PBD}_{(m,k)}$ is the bit dropping probability
conditioned on bit arrival. Since our target bit dropping
probability $\text{PBD}_{(m,k)}\ll 1$,
$\overline{T}_{(m,k)}(\Omega)=\lim\sup_{T\rightarrow
\infty}\frac{1}{T}\sum\nolimits_{t=1}^T\mathbb{E}[\frac{Q_{(m,k)}}{\lambda_{(m,k)}}]\approx\lim\sup_{T\rightarrow
\infty}\frac{1}{T}\sum\nolimits_{t=1}^T\mathbb{E}[\frac{Q_{(m,k)}}{\lambda_{(m,k)}(1-\text{PBD}_{(m,k)})}]$.}
of the average delay of user $k$. When
$f(Q_{(m,k)})=1_{\{Q_{(m,k)}\geq N_Q\}}$ and $A_{(m,k)}$ follows the
bernoulli process, $\overline{T}_{(m,k)}(\Omega)$ is the {\em bit
dropping probability} (conditioned on bit arrival). Note that, the
$MK$ queues in the celluar system are coupled together via the
control policy $\Omega$. In this paper, we seek to find an optimal
stationary control policy $\Omega$ to minimize the average cost in
(\ref{eq:T_single}). Specifically, we have:
\begin{Prob}[Delay Optimal Multi-cell Control Problem]\label{prob:delay} \footnote{In fact, the proposed solution framework
can be easily extended to deal with a more general QoS based
optimization. For example, say we minimize the average delay subject
to the constraints on average data rate:
$\overline{R}_{(m,k)}(\Omega)=\lim\sup_{T\rightarrow
\infty}\frac{1}{T}\sum\nolimits_{t=1}^T\mathbb{E}[R_{(m,k)}(t)]\geq
R_T^k$. The Lagrangian of such constrained optimization is:
$\min_{\Omega}J_{\beta}^{\Omega}=\sum\nolimits_{m,k}\left[
\beta_{(m,k)}\overline{T}_{(m,k)}(\Omega)+\mu_{(m,k)}\overline{R}_{(m,k)}(\Omega)\right]=\lim\sup_{T\rightarrow
\infty}\frac{1}{T}\sum\nolimits_{t=1}^T\mathbb{E}^{\Omega}
[g_{\mu}(\boldsymbol{\chi}(t),\Omega(\boldsymbol{\chi}(t)))]$, where
$g_{\mu}(\boldsymbol{\chi}(t),\Omega(\boldsymbol{\chi}(t)))=\sum_{m,k}\beta_{(m,k)}f(Q_{(m,k)})+\mu_{(m,k)}R_{(m,k)}$,
and $\mu_{(m,k)}$ is the Lagrange multiplier corresponding to the
QoS constraint $\overline{R}_{(m,k)}(\Omega)\geq R_T^k$. Note that
it has the same form as (\ref{eq:problem}) and the proposed solution
framework can be applied to the QoS constrained problem as well.}
For some positive constants
$\boldsymbol{\beta}=\{\beta_{(m,k)},,\forall
k\in\mathcal{K}_m,\forall m\}$, finding a stationary control policy
$\Omega$ that minimizes:
\begin{eqnarray}
\label{eq:problem}
\min_{\Omega}J_{\beta}^{\Omega}&=&\sum\nolimits_{m,k}\beta_{(m,k)}\overline{T}_{(m,k)}(\Omega)\\
&=&\lim\sup_{T\rightarrow
\infty}\frac{1}{T}\sum\nolimits_{t=1}^T\mathbb{E}^{\Omega}[g(\boldsymbol{\chi}(t),\Omega(\boldsymbol{\chi}(t)))]\nonumber
\end{eqnarray}
where
$g(\boldsymbol{\chi}(t),\Omega(\boldsymbol{\chi}(t))=\sum_{m,k}\beta_{(m,k)}f(Q_{(m,k)})$
is the per-slot cost, and $\mathbb{E}^{\Omega}$ denotes the
expectation w.r.t. the induced measure (induced by the control
policy $\Omega$ and the transition kernel in (\ref{eq:sys_tran})).
The positive constants $\boldsymbol{\beta}$ indicate the relative
importance of the users and for a given $\boldsymbol{\beta}$, the
solution to (\ref{eq:problem}) corresponds to a Pareto optimal point
of the multi-objective optimization problem given by $\min_{\Omega}
\overline{T}_{(m,k)}(\Omega), \forall m,k$. Moreover, a control
policy $\Omega^*$ is called Pareto optimal if for any control policy
$\Omega^{\prime}\neq\Omega^*$ such that
$\overline{T}_{(m,k)}(\Omega^{\prime})\leq\overline{T}_{(m,k)}(\Omega^*),
\forall m,k$, it implies that
$\overline{T}_{(m,k)}(\Omega^{\prime})=\overline{T}_{(m,k)}(\Omega^*),
\forall m,k$. In other words, we cannot reduce
$\overline{T}_{(m,k1)}$ without increasing other component (say
$\overline{T}_{(m,k2)}$) at Pareto optimal control
$\Omega^*$\cite{Boyd:2004}.
\end{Prob}

\section{General Solution to the Delay Optimal
Problem}\label{sec:opt_solution} In this section, we will show that
the delay optimal problem \ref{prob:delay} can be modeled as an
infinite horizon average cost POMDP, which is a very difficult
problem. By exploiting the special structure, we shall derive an
{\em equivalent Bellman equation} to solve the POMDP problem.

\subsection{Preliminary on MDP and POMDP}
An infinite horizon average cost MDP can be characterized by a tuple
of four objects:
$\{\mathbb{S},\mathbb{A},\Pr\{s^{\prime}|s,a\},g(s,a)\}$, where
$\mathbb{S}$ is a finite set of states and $\mathbb{A}$ is the
action space. $\Pr\{s^{\prime}|s,a\}$ is the transition probability
from state $s$ to $s^{\prime}$, given that the action
$a\in\mathbb{A}$ is taken. $g(s,a)$ is the per-slot cost function.
The objective is to find the optimal policy $\mathbf{a}=\{a(s)\}$ so
as to minimize the average per-slot cost $\theta$ as:
\begin{equation}
\label{eq:per-theta} \theta=\min_{\mathbf{a}}\lim_{T\rightarrow
\infty}\sup\frac{1}{T}\sum\nolimits_{t=1}^T\mathbb{E}^{\mathbf{a}}[g(s(t),a(s(t)))]
\end{equation}

If the policy space consists of {\em unichain policies} and the
associated induced Markov chain is irreducible, it is well known
that there exist a unique $\theta$ for each starting
state\cite{Cao:2007,Bertsekas:2007}. Furthermore, the optimal
control policy $\mathbf{a}$ can be obtained by the following Bellman
equation.
\begin{equation}
\label{eq:bellman}
V(s)+\theta=\min_{a(s)}\left\{g(s,a(s))+\sum\nolimits_{s^{\prime}}\Pr\{s^{\prime}|s,a(s)\}V(s^{\prime}))\right\}
\end{equation}
where $V(s)$ is called the value function. General offline
solutions, {\em value} or {\em policy iteration}, can be used to
find the value function $V(s)$ iteratively, as well as the optimal
policy\cite{Bertsekas:2007}.

%

POMDP is an extension of MDP when the control agent does not have
direct observation of the entire system state (and hence it is
called ``partially observed MDP''). Specifically, an infinite
horizon average cost POMDP can be characterized by a tuple
\cite{Meuleau:1999,POMDP:1998}:
$\{\mathbb{S},\mathbb{A},\Pr\{s^{\prime}|s,a\},g(s,a),\mathbb{O},O(z,s,a)\}$,
where $\{\mathbb{S},\mathbb{A},P(s^{\prime}|s,a),g(s,a)\}$
characterize a MDP and $\mathbb{O}$ is a finite set of observations.
$O(z,s,a)$ is the observation function, which gives the probability
(or stochastic relationship) between the partial observation $z$,
the actual system state $s$ and the control action $a$.
Specifically, $O(z,s,a)$ is the probability of getting a partial
observation ``$z$'' given that the current system state is $s$ and
the action $a$ was taken in the previous slot. A PODMP is a MDP
where current system state and the actions are based on the
observation $z$. The objective is to find the optimal policy
$\mathbf{a}=\{a(z)\}$ so as to minimize the average per-slot cost
$\theta$ in (\ref{eq:per-theta}). However, in general, it is a {\em
NP-hard} problem and there are various approximation solutions
proposed based on the special structure of the studied
problems\cite{POMDP:survey}.

\subsection{Equivalent Bellman Equation and Optimal Control Policy}
In this subsection, we shall first illustrate that the optimization
problem \ref{prob:delay} is an infinite horizon average cost POMDP.
We shall then exploit some special problem structure to simplify the
complexity and derive an {\em equivalent Bellman equation} to solve
the problem. For instance, in the delay optimal problem
\ref{prob:delay}, the ICI management control policy
$\Omega_{\mathbf{p}}$ is adaptive to the QSI $\mathbf{Q}$, while the
user scheduling policy $\Omega_{\mathbf{s}}$ is adaptive to the
complete system state $\{\mathbf{Q},\mathbf{H}\}$. Therefore, the
optimal control policy $ \Omega^*$ cannot be obtained by solving a
standard Bellman equation from conventional MDP\footnote{The policy
will be a function of the complete system state by solving a
standard bellman equation.}. In fact, problem \ref{prob:delay} is a
POMDP with the following specification.
\begin{itemize}
\item{\bf State Space:}  The system state is the global QSI and CSI
$\boldsymbol{\chi}=\{\mathbf{Q},\mathbf{H}\}\in\{\mathcal{Q},\mathcal{H}\}$.
\item{\bf Action Space:} The action is ICI management pattern and user
scheduling
$\{\mathbf{p},\mathbf{s}\}\in\{\mathcal{P},\mathcal{S}\}$.
\item{\bf Transition Kernel:} The transition probability
$\Pr\{\boldsymbol{\chi}^{\prime}|\boldsymbol{\chi},\mathbf{p},\mathbf{s}\}$
is given in (\ref{eq:sys_tran}).
\item{\bf Per-Slot Cost Function:} The per-slot cost function is
$g(\boldsymbol{\chi},\mathbf{p},\mathbf{s})=\sum_{m,k}\beta_{(m,k)}f(Q_{(m,k)})$.
\item{\bf Observation:} The observation for ICI management control
policy is global QSI, i.e., $z_{\mathbf{p}}=\mathbf{Q}$, while the
observation for User scheduling policy is the complete system state,
i.e., $z_{\mathbf{s}}=\boldsymbol{\chi}$.
\item{\bf Observation Function:} The observation function for ICI management control
policy is
$O_{\mathbf{p}}(z_{\mathbf{p}},\boldsymbol{\chi},\mathbf{p},\mathbf{s})=1$,
if $z_{\mathbf{p}}=\mathbf{Q}$, otherwise 0. Furthermore the
observation function for user scheduling policy is
$O_{\mathbf{s}}(z_{\mathbf{s}},\boldsymbol{\chi},\mathbf{p},\mathbf{s})=1$,
if $z_{\mathbf{s}}=\boldsymbol{\chi}$, otherwise 0.
\end{itemize}

While POMDP is a very difficult problem in general, we shall utilize
the notion of {\em action partitioning} in our problem to
substantially simplify the problem. We first define {\em partitioned
actions} below.
\begin{Def}[Partitioned Actions]
\label{def:partitioned action} Given a control policy $\Omega$, we
define
$\Omega(\mathbf{Q})=\{(\mathbf{p},\mathbf{s})=\Omega(\boldsymbol{\chi}):\boldsymbol{\chi}=(\mathbf{Q},\mathbf{H})\forall\mathbf{H}\in\mathcal{H}\}$
as the collection of actions under a given $\mathbf{Q}$ for all
possible $\mathbf{H}\in\mathcal{H}$. The complete policy $\Omega$ is
therefore equal to the union of all partitioned actions, i.e.,
$\Omega=\bigcup_{\mathbf{Q}}\Omega(\mathbf{Q})$. ~\hfill \IEEEQED
\end{Def}

Based on the action partitioning, we can transform the POMDP problem
into a regular infinite-horizon average cost MDP. Furthermore, the
optimal control policy $\Omega^*$ can be obtained by solving an {\em
equivalent Bellman equation} which is summarized in the theorem
below.
\begin{Thm}[Equivalent Bellman Equation]
\label{Thm:MDP_cond} The optimal control policy $\Omega^*=
(\Omega_{\mathbf{p}}^*, \Omega_{\mathbf{s}}^*)$ in problem
\ref{prob:delay} can be obtained by solving the {\em equivalent
Bellman equation} given by:
\begin{equation}
\label{eq:bellman_cond} V(\mathbf{Q})+\theta =
\min_{\Omega(\mathbf{Q})}\Big[\hat{g}(\mathbf{Q},\Omega(\mathbf{Q}))+\sum\limits_{\mathbf{Q}^{\prime}}\Pr\{\mathbf{Q}^{\prime}|\mathbf{Q},\Omega(\mathbf{Q})\}V(\mathbf{Q}^{\prime})\Big]
\end{equation}
where
$\hat{g}(\mathbf{Q},\Omega(\mathbf{Q}))=\sum_{m,k}\beta_{(m,k)}f(Q_{(m,k)})$
is the per-slot cost function, and the transition kernel is given by
$\Pr\{\mathbf{Q}^{\prime}|\mathbf{Q},\Omega(\mathbf{Q})\}=\mathbb{E}_{\mathbf{H}}\left[\Pr\{\mathbf{Q}^{\prime}|\mathbf{Q},\mathbf{H},\Omega(\boldsymbol{\chi})\}\right]$,
where
$\Pr\{\mathbf{Q}^{\prime}|\mathbf{Q},\mathbf{H},\Omega(\boldsymbol{\chi})\}$
is given by
\begin{equation}\begin{array}{l}
\Pr\{\mathbf{Q}^{\prime}|\mathbf{Q},\mathbf{H},\Omega(\boldsymbol{\chi})\}=\\
\quad\quad\left\{\begin{array}{ll} \Pr\{\mathbf{A}\} & \text{if
$\mathbf{Q}^{\prime}=\Big[\big(\mathbf{Q}-\mathbf{U}\big)^++\mathbf{A}\Big]_{\bigwedge N_Q}$}\\
0 & \text{otherwise}
\end{array}
\right.
\end{array}
\end{equation}
where $\mathbf{U}=\{\mathbf{U}_m\}_{m=1}^M$, and
$\mathbf{U}_m=\{U_{(m,k)}\}_{k=1}^K$, and $U_{(m,k)}=
R_{(m,k)}(\boldsymbol{\chi},\Omega(\boldsymbol{\chi}))\tau$ for
$k\in\mathcal{K}_m$. Suppose
$\Omega^*(\mathbf{Q})=\{\mathbf{p}^*({\mathbf{Q}}),\bigcup_{\mathbf{H}}\mathbf{s}^*(\mathbf{Q},\mathbf{H})\}$
is a solution that solves the Bellman equation in
(\ref{eq:bellman_cond}), the optimal control policy for the original
Problem \ref{prob:delay} is given by:
$\Omega_{\mathbf{p}}^*=\bigcup_{\mathbf{Q}}\{\mathbf{p}^*(\mathbf{Q})\}$
and
$\Omega_{\mathbf{s}}^*=\bigcup_{\mathbf{Q},\mathbf{H}}\{\mathbf{s}^*(\mathbf{Q},\mathbf{H})\}$.
The value function $V(\mathbf{Q})$ that solves
(\ref{eq:bellman_cond}) is a component-wise monotonic increasing
function. ~\hfill \IEEEQED
\end{Thm}

\begin{proof}
Please refer to Appendix A. 
\end{proof}

Note that solving (\ref{eq:bellman_cond}) will obtain an ICI
management policy $\Omega_{\mathbf{p}}^*$ that is a function of QSI
$\mathbf{Q}$ and a user scheduling policy $\Omega_{\mathbf{s}}^*$
that is a function of the QSI and CSI $\{\mathbf{Q},\mathbf{H}\}$.
We shall illustrate this with a simple example below.

\begin{Exam}
Suppose there are two BSs with equal transmitting power
($P^m_{\max}=P,\forall m$), and there are three ICI management
control patterns in $\mathcal{P}$, given by
$\mathbf{p}_1=\{p_1=1,p_2=0\}$ (BS 1 is active),
$\mathbf{p}_2=\{p_1=0,p_2=1\}$ (BS 2 is active) and
$\mathbf{p}_3=\{p_1=1,p_2=1\}$ (both BSs are active). Assume
deterministic arrival where one bit will always arrive at each slot,
i.e., $\Pr\{A_{(m,k)}=1\}=1$. The number of users served by each BS
is $K=2$. The path loss $L_{(m,k)}^n=1$ for all $\{k,n,m\}$, and the
small scale fading gain is chosen from two values $\{H_{g},H_{b}\}$
with equal probability. As a result, the global CSI state
space\footnote{For the sake of easy discussion, we consider discrete
state space in this example. Yet, the proposed algorithms and
convergence results in the paper work for general continuous state
space as well.} is $\mathcal{H}=\{H_{g},H_{b}\}^{M^2K}$. Note that
the cardinality of CSI state space $\mathcal{H}$ is
$|\mathcal{H}|=2^{M^2K}=256$. Given a realization of the global QSI
$\mathbf{Q}$, the {\em partitioned actions} (following Definition
\ref{def:partitioned action}) is given by:
\begin{equation}
\Omega(\mathbf{Q})=\{\mathbf{p}(\mathbf{Q}),
\mathbf{s}(\mathbf{Q},\mathbf{H}^{(1)}),\cdots,
\mathbf{s}(\mathbf{Q},\mathbf{H}^{(256)})\}
\end{equation}
Using Theorem \ref{Thm:MDP_cond}, the optimal partitioned action
$\Omega^*(\mathbf{Q})$ is given by solving the right hand side (RHS)
of (\ref{eq:bellman_cond}):
\begin{equation}\begin{array}{l}
\Omega^*(\mathbf{Q})=\argmin\limits_{\{\mathbf{p}(\mathbf{Q}),
\{\mathbf{s}(\mathbf{Q},\mathbf{H}^{(i)})\}_{i=1}^{256}\}}\sum_{\mathbf{Q}^{\prime}}\sum_{\mathbf{H}^{(i)}
\in \mathcal{H} }\\
\quad\Big[\Pr\{\mathbf{H}^{(i)}\}
\Pr\{\mathbf{Q}^{\prime}|\mathbf{Q},\mathbf{H}^{(i)},\mathbf{p}(\mathbf{Q}),
\mathbf{s}(\mathbf{Q},\mathbf{H}^{(i)})\}V(\mathbf{Q}^{\prime})\Big]
\end{array}
\end{equation}
where
\begin{equation}\begin{array}{l}
\Pr\{\mathbf{Q}^{\prime}|\mathbf{Q},\mathbf{H}^{(i)},\mathbf{p}(\mathbf{Q}),\mathbf{s}(\mathbf{Q},\mathbf{H}^{(i)})\}=\\
\quad\quad\quad\quad\left\{\begin{array}{ll} 1 & \text{if }
\mathbf{Q}^{\prime} =
\Big[\big(\mathbf{Q}-\mathbf{U}\big)^++\mathbf{1}\Big]_{\bigwedge N_Q}\\
0 & \text{otherwise}
\end{array}\right.
\end{array}
\end{equation}
and $\mathbf{U}=\{U_{(1,1)},U_{(1,2)};U_{(2,1)},U_{(2,2)}\}$ is the
number of departure bits. 
For a given ICI management control
$\mathbf{p}(\mathbf{Q})=\mathbf{p}$, the optimal user scheduling
policy $\{\mathbf{s}^*(\mathbf{Q},\mathbf{H}^{(i)})\}$ is
\begin{equation}
\label{eq:tmp}\begin{array}{l}
\{\mathbf{s}^*(\mathbf{Q},\mathbf{H}^{(i)})\}=
\argmin\limits_{\{\mathbf{s}(\mathbf{Q},\mathbf{H}^{(i)})\}_{i=1}^{256}}\sum_{\mathbf{Q}^{\prime}}
\sum_{\mathbf{H}^{(i)} \in
\mathcal{H}}\\
\quad\Big[\Pr\{\mathbf{H}^{(i)}\}\Pr\{\mathbf{Q}^{\prime}|\mathbf{Q},\mathbf{H}^{(i)},\mathbf{p},
\mathbf{s}(\mathbf{Q},\mathbf{H}^{(i)})\}V(\mathbf{Q}^{\prime})\Big]
\end{array}
\end{equation}
Observe that the RHS of (\ref{eq:tmp}) is a decoupled objective
function w.r.t. the variables
$\{\mathbf{s}(\mathbf{Q},\mathbf{H}^{(i)})\}_{i=1}^{256}$ and hence,
applying standard decomposition theory,
\begin{equation}
\label{eq:opt_s}\begin{array}{l}
\mathbf{s}^*(\mathbf{Q},\mathbf{H}^{(i)})=\\
\quad\argmin\limits_{\mathbf{s}(\mathbf{Q},\mathbf{H}^{(i)})}\sum\limits_{\mathbf{Q}^{\prime}}\Pr\{\mathbf{Q}^{\prime}|\mathbf{Q},\mathbf{H}^{(i)},\mathbf{p},\mathbf{s}(\mathbf{Q},\mathbf{H}^{(i)})\}V(\mathbf{Q}^{\prime})
\end{array}
\end{equation}
As a result, the optimal ICI management control policy
$\mathbf{p}^*(\mathbf{Q})$ is given by:
\begin{equation}
\label{eq:opt_p}\begin{array}{l}
\mathbf{p}^*(\mathbf{Q})=\argmin_{\mathbf{p}(\mathbf{Q})}\sum\nolimits_{\mathbf{Q}^{\prime}}
\sum\nolimits_{\mathbf{H}^{(i)}\in\mathcal{H}}\\
\quad\Big[\Pr\{\mathbf{H}^{(i)}\}\Pr\{\mathbf{Q}^{\prime}|\mathbf{Q},\mathbf{H}^{(i)},\mathbf{p}(\mathbf{Q}),
\mathbf{s}^*(\mathbf{Q},\mathbf{H}^{(i)})\}V(\mathbf{Q}^{\prime})\Big]
\end{array}
\end{equation}
where $\mathbf{s}^*(\mathbf{Q},\mathbf{H}^{(i)})$ given in
(\ref{eq:opt_s}) is the optimal user scheduling policy under the ICI
management control policy $\mathbf{p}(\mathbf{Q})$. Using Theorem
\ref{Thm:MDP_cond}, the optimal ICI management control and user
selection control of the original Problem \ref{prob:delay} for a CSI
realization $\mathbf{H}^{(i)}$ and QSI realization $\mathbf{Q}$ are
given by $\mathbf{p}^*(\mathbf{Q})$ and $\mathbf{s}^*(\mathbf{Q},
\mathbf{H}^{(i)})$ respectively. ~\hfill \IEEEQED

\end{Exam}

\section{Distributive Value Function and $\mathbb{Q}$-factor Online Learning}\label{sec:learning}
The solution in Theorem \ref{Thm:MDP_cond} requires the knowledge of
the value function $V(\mathbf{Q})$. However, obtaining the value
function is not trivial as solving the Bellman equation
(\ref{eq:bellman_cond}) involves solving a very large system of the
nonlinear fixed point equations (corresponding to each realization
of $\mathbf{Q}$ in (\ref{eq:bellman_cond})). Brute-force solution of
$V(\mathbf{Q})$ require huge complexity, centralized implementation
and knowledge of global CSI and QSI at the BSC. This will also
induce huge signaling overhead because the QSI of all the users are
maintained locally at the $M$ BSs. In this section, we shall propose
a decentralized solution via distributive stochastic learning
following the structure as illustrated in Fig.
\ref{fig:learning_structure}. Moreover, we shall prove that the
proposed distributive stochastic learning algorithm will converge
almost-surely.

\subsection{Post-Decision State Framework}
In this section, we first introduce the post-decision state also
used framework, also used in \cite{Thesis:Salodkar} and the
references therein, to lay ground for developing the online learning
algorithm. The post-decision state is defined to be the virtual
system state immediately after making an action but before the new
bits arrive. For example,
$\boldsymbol{\chi}=\{\mathbf{Q},\mathbf{H}\}$ is the state at the
beginning of some time slot (also called the {\em pre-decision
state}), and making an action
$\Omega(\boldsymbol{\chi})=\{\mathbf{p},\mathbf{s}\}$, the
post-decision state immediately after the action is
$\widetilde{\boldsymbol{\chi}}=\{\widetilde{\mathbf{Q}},\mathbf{H}\}$,
where the transition to $\widetilde{\mathbf{Q}}$ is given by
$\widetilde{\mathbf{Q}}=\big(\mathbf{Q}-\mathbf{U}\big)^+$. If new
arrivals $\mathbf{A}$ occur in the post-decision state, and the CSI
changes to $\mathbf{H}^{\prime}$, then the system reaches the next
actual state, i.e., pre-decision state,
$\boldsymbol{\chi}^{\prime}=\{\big[\widetilde{\mathbf{Q}}+\mathbf{A}\big]_{\bigwedge
N_Q},\mathbf{H}^{\prime}\}$.

Using the action partitioning and defining the value function
$\widetilde{\mathbf{V}}$ on post-decision state
$\widetilde{\mathbf{Q}}$ (where pre-decision state is
$\{\mathbf{Q}=\big[\widetilde{\mathbf{Q}}+\mathbf{A}\big]_{\bigwedge
N_Q},\mathbf{H}\}$), $\widetilde{\mathbf{V}}$ will satisfy the
post-decision state Bellman equation\cite{Thesis:Salodkar}
\begin{equation}
\label{eq:bellman_post}\begin{array}{l}
\widetilde{V}(\widetilde{\mathbf{Q}})+\theta=\sum_{\mathbf{A}}\Pr\{\mathbf{A}\}\bigg\{
\min_{\Omega(\mathbf{Q})}\Big[\widetilde{g}(\mathbf{Q},\Omega(\mathbf{Q}))\\
\quad\quad\quad\quad+\sum\nolimits_{\widetilde{\mathbf{Q}}^{\prime}}\Pr\{\widetilde{\mathbf{Q}}^{\prime}|\mathbf{Q},\Omega(\mathbf{Q})\}\widetilde{V}(\widetilde{\mathbf{Q}}^{\prime})
\Big]\bigg\}
\end{array}
\end{equation}
where
$\widetilde{g}(\mathbf{Q},\Omega(\mathbf{Q}))=\sum_{m,k}\beta_{(m,k)}f(Q_{(m,k)})$,
$\Pr\{\widetilde{\mathbf{Q}}^{\prime}|\mathbf{Q},\Omega(\mathbf{Q})\}=\mathbb{E}_{\mathbf{H}}[\Pr\{\widetilde{\mathbf{Q}}^{\prime}|\mathbf{Q},\mathbf{H},\Omega(\mathbf{Q})\}]$,
and $\widetilde{\mathbf{Q}}^{\prime}$ is the next post-decision
state transited from $\mathbf{Q}$. As Theorem \ref{Thm:MDP_cond},
$\widetilde{V}(\widetilde{\mathbf{Q}})$ is also a component-wise
monotonic increasing function. The optimal policy is obtained by
solving the RHS of Bellman equation (\ref{eq:bellman_post}).


\subsection{Distributive User Scheduling Policy on the CSI Time Scale}
To reduce the size of the state space and to decentralize the user
scheduling, we approximate $\widetilde{V}(\widetilde{\mathbf{Q}})$
in (\ref{eq:bellman_post}) by the sum of per-user post-decision
state value function\footnote{Using the linear approximation in
(\ref{eq:linear_value}), we can address the curse of dimensionality
(complexity) as well as facilitate distributive implementation where
each BS could solve for
$\widetilde{V}_{(m,k)}(\widetilde{Q}_{(m,k)})$ based on local CSI
and QSI only.} $\widetilde{V}_{(m,k)}(\widetilde{Q}_{(m,k)})$, i.e.,
\begin{equation}
\label{eq:linear_value}
\widetilde{V}(\widetilde{\mathbf{Q}})\approx\sum\nolimits_{m,k}\widetilde{V}_{(m,k)}(\widetilde{Q}_{(m,k)})
\end{equation}
where $\widetilde{V}_{(m,k)}(\widetilde{Q}_{(m,k)})$ is defined as
the {\em fixed point} of the following per-user fixed point
equation:
\begin{equation}
\label{eq:user_selection}\begin{array}{l}
\widetilde{V}_{(m,k)}(\widetilde{Q}_{(m,k)})+\widetilde{V}_{(m,k)}(\widetilde{Q}_{(m,k)}^I)=\\
\sum\nolimits_{A_{(m,k)}}\Pr\{A_{(m,k)}\}\Big[ \beta_{(m,k)}f(Q_{(m,k)})+\\
\sum\limits_{\widetilde{Q}_{(m,k)}^{\prime}}\Pr\{\widetilde{Q}_{(m,k)}^{\prime}|Q_{(m,k)},s_{(m,k)}=1,\widetilde{\mathbf{p}}_m^I\}\widetilde{V}_{(m,k)}(\widetilde{Q}_{(m,k)}^{\prime})
\Big]
\end{array}
\end{equation}
where $Q_{(m,k)}=\widetilde{Q}_{(m,k)}+A_{(m,k)}$ is the
pre-decision state, $s_{(m,k)}=1$ means that the user $k$ is
scheduled to transmit at BS $m$,
$\widetilde{Q}_{(m,k)}^I\in\{0,\cdots,N_Q\}$ is a reference state
and $\widetilde{\mathbf{p}}_m^I\in\mathcal{P}_m$ is a reference ICI
management pattern (with the BS $m$ active). The per-user value
function $\widetilde{V}_{(m,k)}(\widetilde{Q}_{(m,k)})$ is obtained
by the proposed distributive online learning algorithm (explained in
section \ref{subsec:learning_algorithm}). Note that the state space
for the value function of $\widetilde{V}(\widetilde{\mathbf{Q}})$ is
substantially reduced from $(N_Q+1)^{MK}$ (exponential growth w.r.t
the number of all mobile users $MK$) to $MK(N_Q+1)$ (linear growth
w.r.t the number of all mobile users).

\begin{Cor}[Decentralized User Scheduling Actions]\label{cor:dstr_policy} Using the linear
approximation in (\ref{eq:linear_value}), the user scheduling action
of BS $m\in\mathcal{M}_{\mathbf{p}}$ under any given ICI management
pattern $\mathbf{p}$ (obtained by solving the RHS of Bellman
equation (\ref{eq:bellman_post})) is given by:
\begin{equation}
\label{eq:sm} \mathbf{s}_{m}=\{s_{k^*}=1, s_{(m,k)}=0, \forall k\neq
k^*\text{ and } k,k^*\in\mathcal{K}_m\}
\end{equation}
where
$k^*=\arg\max_{k\in\mathcal{K}_m}\widetilde{\delta}_{(m,k)}(Q_{(m,k)})$,
and
$\widetilde{\delta}_{(m,k)}(Q_{(m,k)})=\widetilde{V}_{(m,k)}(Q_{(m,k)})-\widetilde{V}_{(m,k)}((Q_{(m,k)}-U_{(m,k)})^+)$\footnote{
Note that $\widetilde{\delta}_{(m,k)}(0)=0,\forall k$, and hence the
users with empty buffer will not be scheduled and the activated BS
$m$ will serve the users with non-empty buffer (the chance for the
buffer of all $K$ users being empty at a given slot is very
small).}. $U_{(m,k)}=\log_2\left(1+\frac{\xi
\phi_{(m,k)}}{\varphi_{(m,k)}}\right)\tau$, where
$\varphi_{(m,k)}=\sum\limits_{n\neq m,n\in\mathcal{M}_{\mathbf{p}}}
P^n_{\max}H^n_{(m,k)}L_{(m,k)}^n+N_0W$ is the power sum of
interference and noise, and
$\phi_{(m,k)}=P^m_{\max}H^m_{(m,k)}L_{(m,k)}^m$ is the signal power.
~\hfill\IEEEQED
\end{Cor}
\begin{proof}
Please refer to Appendix B. 
\end{proof}

\begin{Rem}[Structure of the User Scheduling Actions]
The user scheduling action in (\ref{eq:sm}) is both function of
local CSI and QSI. Specifically, the number of bits to be delivered
$U_{(m,k)}$ is controlled by the local CSI $\mathbf{H}_{(m,k)}$, and
local QSI $Q_{(m,k)}$ will determine
$\widetilde{\delta}_{(m,k)}(Q_{(m,k)})$. Each user estimates
$\varphi_{(m,k)}$ and $\phi_{(m,k)}$ in the preamble phase, and
sends $U_{(m,k)}$ to the associated BS $m$ according to the process
as indicated in Fig.\ref{fig:learning_structure}. ~\hfill\IEEEQED
\end{Rem}

\begin{figure}
 \begin{center}
  \resizebox{9cm}{!}{\includegraphics{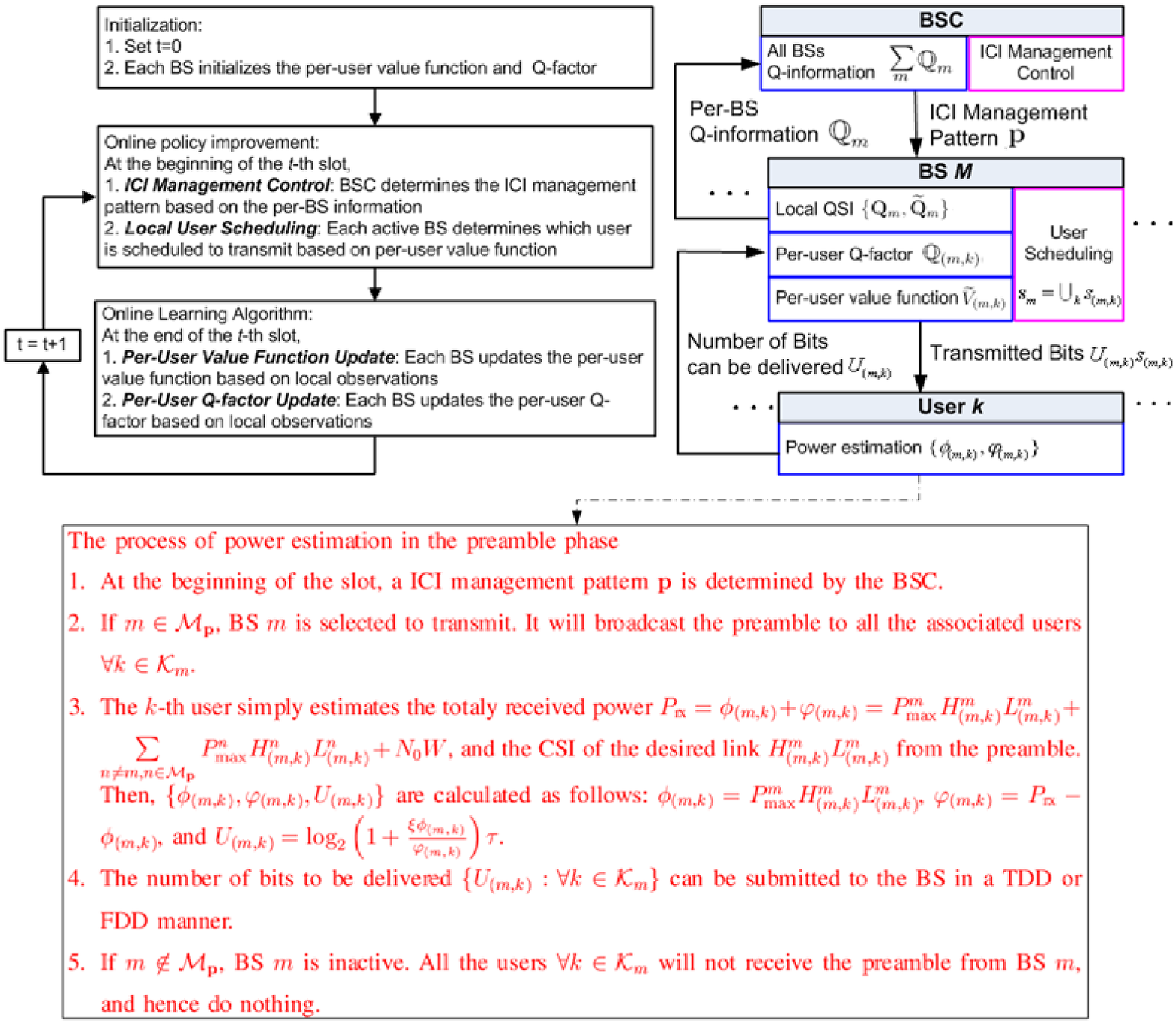}}
 \end{center}
    \caption{The system procedure for distributive per-user value function and per-user $\mathbb{Q}$-factor online learning algorithm.}
    \label{fig:learning_structure}
\end{figure}
\subsection{ICI Management Control Policy on the QSI Time Scale}
To determine the ICI management control policy, we define the
$\mathbb{Q}$-factor as follows \cite{Cao:2007}:
\begin{equation}
\label{eq:q_factor_orgn}\begin{array}{l}
\mathbb{Q}(\mathbf{Q},\mathbf{p})=\sum\nolimits_{m,k}\beta_{(m,k)}f(Q_{(m,k)})+\\
\quad\quad\quad\sum\nolimits_{\mathbf{Q}^{\prime}}\Pr\{\mathbf{Q}^{\prime}|\mathbf{Q},\mathbf{p}\}
\min_{\mathbf{p}^{\prime}}\mathbb{Q}(\mathbf{Q}^{\prime},\mathbf{p}^{\prime})-\theta
\end{array}
\end{equation}
where $\Pr\{\mathbf{Q}^{\prime}|\mathbf{Q},\mathbf{p}\}$ is the
transition probability from current QSI $\mathbf{Q}$ to
$\mathbf{Q}^{\prime}$, given current action $\mathbf{p}$, and
$\theta$ is a constant. Note that the $\mathbb{Q}$-factor
$\mathbb{Q}(\mathbf{Q},\mathbf{p})$ represents the potential cost of
applying a control action $\mathbf{p}$ at the current QSI
$\mathbf{Q}$ and applying the action
$\arg\min_{\mathbf{p}^{\prime}}\mathbb{Q}(\mathbf{Q}^{\prime},\mathbf{p}^{\prime})$
for any system state $\mathbf{Q}^{\prime}$ in the future. Similar to
(\ref{eq:linear_value}), we approximate the $\mathbb{Q}$-factor in
(\ref{eq:q_factor_orgn}) with a sum of per-user $\mathbb{Q}$-factor,
i.e,
\begin{equation}
\label{eq:q_approx}
\mathbb{Q}(\mathbf{Q},\mathbf{p})\approx\sum\nolimits_{m,k}\mathbb{Q}_{(m,k)}(Q_{(m,k)},\mathbf{p})
\end{equation}
where $\mathbb{Q}_{(m,k)}$ is defined as the {\em fixed point} of
the following per-user fixed point equation:
\begin{equation}
\label{eq:Q_linear}\begin{array}{l}
\mathbb{Q}_{(m,k)}(Q_{(m,k)},\mathbf{p})=\\
\beta_{(m,k)}f(Q_{(m,k)})-
\mathbb{Q}_{(m,k)}(Q_{(m,k)}^I,\mathbf{p}^I_m)+\sum\limits_{Q_{(m,k)}^{\prime}}\\
\Big[\Pr\{Q_{(m,k)}^{\prime}|Q_{(m,k)},
s_{(m,k)}=1,\mathbf{p}\}\min\limits_{\mathbf{p}^{\prime}}\mathbb{Q}_{(m,k)}(Q_{(m,k)}^{\prime},
\mathbf{p}^{\prime})\Big]
\end{array}
\end{equation}
where
$\Pr\{Q_{(m,k)}^{\prime}|Q_{(m,k)},s_{(m,k)}=1,\mathbf{p}\}=\mathbb{E}_{\mathbf{H}_{(m,k)}}[\Pr\{Q_{(m,k)}^{\prime}|Q_{(m,k)},s_{(m,k)}=1,\mathbf{H}_{(m,k)},\mathbf{p}\}]$.
$Q_{(m,k)}^I\in\{0,\cdots,N_Q\}$ is a reference state and
$\mathbf{p}^I_m\in\mathcal{P}$ is a reference ICI management control
pattern. The per-user $\mathbb{Q}$-factor $\mathbb{Q}_{(m,k)}$ is
obtained by the proposed distributive online learning algorithm
(explained in section \ref{subsec:learning_algorithm}). The BSC
collects the per-BS $\mathbb{Q}$-information
$\mathbb{Q}_m^t(\mathbf{p})=\sum_{(m,k)}\mathbb{Q}_{(m,k)}^t(Q_{(m,k)}^t,\mathbf{p})$
at the beginning of slot $t$, and the ICI management control policy
is given by:
\begin{equation}
\label{eq:learn_pattern_pol}
\mathbf{p}^t=\argmin\nolimits_{\mathbf{p}}\sum\nolimits_m\mathbb{Q}_m^t(\mathbf{p})
\end{equation}

In order to reduce the communication overhead between the $M$ BSs
and the BSC, we could further partition the local QSI space into $N$
regions\footnote{For example, one possible criteria is to partition
the local QSI space so that the probability of $\mathbf{Q}_{m}$
belonging to any region is the same (uniform probability
partitioning).} ($\mathcal{Q}_m=\bigcup_{n=1}^N \mathcal{R}_n$) as
illustrated in Fig. \ref{fig:pkt_region}. At the beginning of the
$t$-th slot, the $m$-th BS will update the BSC of the per-BS
$\mathbb{Q}$-information if its QSI state belongs to a new region.
Hence, the per-BS $\mathbb{Q}$-information at the BSC is updated
according to the following dynamics:
\begin{equation}
\label{eq:BSC_q_update}\begin{array}{l}
\mathbb{Q}_m^t(\mathbf{p})=\\
\left\{\begin{array}{ll}
\sum\limits_{m,k}\mathbb{Q}_{(m,k)}^t(Q_{(m,k)}^t,\mathbf{p}) &
\text{if $\mathbf{Q}_m^t\in\mathcal{R}_n,\mathbf{Q}_m^{t-1}\not\in
\mathcal{R}_n$}\\
\mathbb{Q}_m^{t-1}(\mathbf{p}) & \text{otherwise}
\end{array}
\right.
\end{array}
\end{equation}

\begin{Rem}[Communication Overhead]
The communication overhead between the $M$ BS and the BSC is reduced
from $O((N_Q+1)^{MK}+(N_H)^{M^2K})$ (exponential growth w.r.t the
number of users $K$) to $O(M(\alpha)^{|\mathcal{P}|})$ for some
constant $\alpha$ (O(1) w.r.t. $K$), where $N_H$ is the cardinality
of the CSI state space for one link.  ~\hfill \IEEEQED
\end{Rem}

\begin{figure}
 \begin{center}
  \resizebox{9cm}{!}{\includegraphics{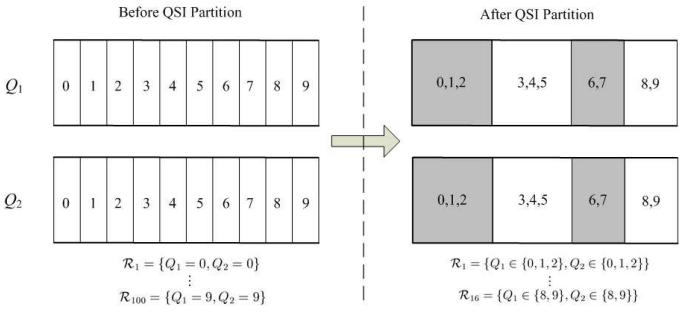}}
 \end{center}
    \caption{Illustration of one possible way of local QSI partition. There are $K=2$ users with buffer size $N_Q=9$, where each user's
QSI is partitioned into 4 regions, given by
$\big\{\{0,1,2\};\{3,4,5\};\{6,7\};\{8,9\}\big\}$. Note that the
number of local QSI regions for one BS is largely reduced from
$(N_Q+1)^K=100$ $\big(\mathcal{R}_1 =
\{Q_1=0,Q_2=0\},\cdots,\mathcal{R}_{100} = \{Q_1=9,Q_2=9\}\big)$ to
$4^K=16$ $\big(\mathcal{R}_1 =
\{Q_1\in\{0,1,2\},Q_2\in\{0,1,2\}\},\cdots, \mathcal{R}_{16} =
\{Q_1\in\{8,9\},Q_2\in\{8,9\}\} \big)$  after partition. }
    \label{fig:pkt_region}
\end{figure}

\subsection{Online Per-User Value Function and Per-User $\mathbb{Q}$-factor Learning Algorithm}
\label{subsec:learning_algorithm} The system procedure for
distributive online learning is given below:
\begin{itemize}
\item{\bf Initialization}: Each BS initiates the per-user value function and $\mathbb{Q}$-factor for
its $K$ users, denoted as $\{\widetilde{V}_{(m,k)}^0\}$ and
$\{\mathbb{Q}_{(m,k)}^0\}$, where
$\widetilde{V}_{(m,k)}^0(Q^{\prime}_{(m,k)})>\widetilde{V}_{(m,k)}^0(Q_{(m,k)}),
\forall Q^{\prime}_{(m,k)}>Q_{(m,k)}$.
\item{\bf ICI Management Control}: At the beginning of the
$t$-th slot, the BSC updates the $\mathbb{Q}$-information
$\mathbb{Q}_m^t(\mathbf{p})$ as (\ref{eq:BSC_q_update}) and
determines the ICI management pattern as
(\ref{eq:learn_pattern_pol}).
\item{\bf User Scheduling}:
If $m\in\mathcal{M}_{\mathbf{p}^t}$, BS $m$ is selected to transmit.
The user scheduling policy is determined according to (\ref{eq:sm}).
\item{\bf Local Per-user Value Function and Per-user $\mathbb{Q}$-factor Update}:
Based on the current observations, each of the $M$ BSs updates the
per-user value function $\widetilde{V}_{(m,k)}$ and the per-user
$\mathbb{Q}$-factor $\mathbb{Q}_{(m,k)}$ according to Algorithm
\ref{alg:learning}.
\end{itemize}

Fig. \ref{fig:learning_structure} illustrates the above procedure by
a flowchart. The algorithm for the per-user value function and
per-user $\mathbb{Q}$-factor update is given below:
\begin{Alg}[Online Learning Algorithm] \label{alg:learning}
Let $\widetilde{\mathbf{Q}}_m$ and $\mathbf{Q}_m$ be the current
observation of post-decision and pre-decision states respectively,
$\mathbf{A}_m$ be the current observation of new arrival,
$\{\mathbf{H}_{(m,k)}\}_{k=1}^K$ be the current observation of the
local CSI, and $\mathbf{p}$ is the realization of the ICI management
control pattern. The online learning algorithm for user
$k\in\mathcal{K}_m$ is given by
\begin{equation}
\label{eq:learn_value_f}\begin{array}{l}
\widetilde{V}_{(m,k)}^{t+1}(\widetilde{Q}_{(m,k)})=\\
\left\{\begin{array}{ll}
\widetilde{V}_{(m,k)}^t(\widetilde{Q}_{(m,k)})+\gamma(t)
\Bigl[\beta_{(m,k)}f(\widetilde{Q}_{(m,k)}+& \\
A_{(m,k)})+\widetilde{V}_{(m,k)}^{t}(\widetilde{Q}_{(m,k)}+A_{(m,k)}-U_{(m,k)})
&\text{if } \mathbf{p} =
\widetilde{\mathbf{p}}_m^I \\
\quad\quad-\widetilde{V}_{(m,k)}^{t}(\widetilde{Q}_{(m,k)}^I)-\widetilde{V}_{(m,k)}^{t}(\widetilde{Q}_{(m,k)})\Bigr]
&\\
\widetilde{V}_{(m,k)}^t(\widetilde{Q}_{(m,k)}) &
\text{otherwise}
\end{array}\right.\end{array}
\end{equation}
\begin{equation}
\label{eq:learn_q_factor}\begin{array}{l}
\mathbb{Q}_{(m,k)}^{t+1}(Q_{(m,k)},\mathbf{p})=\mathbb{Q}_{(m,k)}^{t}(Q_{(m,k)},\mathbf{p})+
\gamma(t)\Bigl[\beta_{(m,k)}\\
\cdot f(Q_{(m,k)})-\mathbb{Q}_{(m,k)}^t(Q_{(m,k)}^I,\mathbf{p}_m^I)-\mathbb{Q}_{(m,k)}^{t}(Q_{(m,k)},\mathbf{p})\\
+\min_{\mathbf{p}^{\prime}}\mathbb{Q}_{(m,k)}^t(Q_{(m,k)}-U_{(m,k)}+A_{(m,k)},\mathbf{p}^{\prime})\Bigr]
\end{array}
\end{equation}
where $U_{(m,k)}$ is the number of bits to be delivered for user $k$
(given in Corollary \ref{cor:dstr_policy} and depends indirectly on
the local CSI observations $\mathbf{H}_{(m,k)}$),
$\{\widetilde{Q}_{(m,k)}^I,\widetilde{\mathbf{p}}_m^I\}$ and
$\{Q_{(m,k)}^I,\mathbf{p}_m^I\}$ are the reference state and
reference ICI management pattern for the value function
$\widetilde{V}_{(m,k)}$ in (\ref{eq:user_selection}) and
$\mathbb{Q}$-factor $\mathbb{Q}_{(m,k)}$ in (\ref{eq:Q_linear})
respectively. $\gamma(n)$ is diminishing positive step size sequence
satisfying $\sum_n\gamma(n)=\infty,\sum_n\gamma^2(n)<\infty$.
~\hfill\IEEEQED
\end{Alg}

\begin{Rem}[Complexity of the Learning Algorithm]
The proposed learning scheme only requires the observations of the
local QSI $\widetilde{\mathbf{Q}}_m$ and $\mathbf{Q}_m$.
Furthermore, each users only need to feedback $U_{(m,k)}$ instead of
the local CSI $\mathbf{H}_m$, which is of similar feedback loading
compared with HSDPA systems. ~\hfill \IEEEQED
\end{Rem}


\subsection{Convergence Analysis}
In this section we will establish the convergence proof of the
proposed per-user learning algorithm \ref{alg:learning}. We first
define a mapping on the post-decision state $\widetilde{Q}_{(m,k)}$
as
\begin{equation}
\label{eq:T_{(m,k)}}\begin{array}{l}
T_{(m,k)}(\widetilde{\mathbf{V}}_{(m,k)},\widetilde{Q}_{(m,k)})=\widetilde{g}_{(m,k)}(\widetilde{Q}_{(m,k)})+\\
\sum\limits_{\widetilde{Q}_{(m,k)}^{\prime}}\Pr\{\widetilde{Q}_{(m,k)}^{\prime}|\widetilde{Q}_{(m,k)},s_{(m,k)}=1,\widetilde{\mathbf{p}}_m^I\}\widetilde{V}_{(m,k)}(\widetilde{Q}_{(m,k)}^{\prime})
\end{array}
\end{equation}
where $Q_{(m,k)}=\widetilde{Q}_{(m,k)}+A_{(m,k)}$ is the
pre-decision state, $\widetilde{g}_{(m,k)}(\widetilde{Q}_{(m,k)})=
\mathbb{E}_{A_{(m,k)}}\big[\beta_{(m,k)}f(\widetilde{Q}_{(m,k)}+A_{(m,k)})\big]$,
and
$\Pr\{\widetilde{Q}_{(m,k)}^{\prime}|\widetilde{Q}_{(m,k)},s_{(m,k)}=1,\widetilde{\mathbf{p}}_m^I\}=
\mathbb{E}_{\mathbf{H}_{(m,k)},A_{(m,k)}}[\Pr\{\widetilde{Q}_{(m,k)}^{\prime}|Q_{(m,k)},\mathbf{H}_{(m,k)},
s_{(m,k)}=1,\widetilde{\mathbf{p}}_m^I\}]$. The vector form of the
mapping is given by:
\begin{equation}
\label{eq:Tk_vector}
\mathbf{T}_{(m,k)}(\widetilde{\mathbf{V}}_{(m,k)})=\widetilde{\mathbf{g}}_{(m,k)}+\mathbf{P}_{(m,k)}\widetilde{\mathbf{V}}_m
\end{equation}
where $\mathbf{P}_{(m,k)}$ is $(N_Q+1)\times(N_Q+1)$ transition
matrix for the post-decision state queue of the user $k$. and
$\widetilde{\mathbf{V}}_{(m,k)}$ are $(N_Q+1)\times1$ vectors.
Specifically, we have the following lemma for the per-user value
function learning in (\ref{eq:learn_value_f}).
\begin{Lem}[Convergence of Per-User Value Function]\label{lem:converge_value}
The update of the per-user value function
$\widetilde{\mathbf{V}}_{(m,k)}^t$ will converge almost-surely in
the proposed learning algorithm \ref{alg:learning}, i.e.,
$\lim_{t\rightarrow\infty}\widetilde{\mathbf{V}}_{(m,k)}^t=\widetilde{\mathbf{V}}_{(m,k)}^{\infty},\forall
k,m$, and $\widetilde{V}_{(m,k)}^{\infty}(\widetilde{Q}_{(m,k)})$ is
a monotonic increasing function satisfying:
\begin{equation}
\label{eq:con_value}
\widetilde{\mathbf{V}}_{(m,k)}^{\infty}+\widetilde{V}_{(m,k)}^{\infty}(\widetilde{Q}_{(m,k)}^I)\mathbf{e}=\mathbf{T}_{(m,k)}(\widetilde{\mathbf{V}}_{(m,k)}^{\infty})
\end{equation}
\end{Lem}
\begin{proof}
Please refer to Appendix C. 
\end{proof}

Note that (\ref{eq:con_value}) is equivalent to the per-user fixed
point equation in (\ref{eq:user_selection}). This result illustrates
that the proposed online distributive learning in
(\ref{eq:learn_value_f}) can converge to the target per-user fixed
point solution in (\ref{eq:user_selection}). We define a mapping for
the per-user $\mathbb{Q}$-factor $\mathbb{Q}_{(m,k)}$ as
\begin{equation}
\label{eq:Q_map}\begin{array}{l}
T_{(m,k)}^{\mathbb{Q}}(\mathbb{Q}_{(m,k)},Q_{(m,k)},\mathbf{p})=
\beta_{(m,k)}f(Q_{(m,k)})+\sum\limits_{Q_{(m,k)}^{\prime}}\\
\Big[\Pr\{Q_{(m,k)}^{\prime}|Q_{(m,k)},s_{(m,k)}=1,
\mathbf{p}\}\min\limits_{\mathbf{p}^{\prime}}\mathbb{Q}_{(m,k)}(Q_{(m,k)}^{\prime},\mathbf{p}^{\prime})\Big]
\end{array}
\end{equation}
Specifically, we have following lemma for the $\mathbb{Q}$-factor
online learning in (\ref{eq:learn_q_factor}).
\begin{Lem}[Convergence of the Per-User $\mathbb{Q}$-factor]\label{lem:q_converge} The update of per-user
$\mathbb{Q}$-factor $\mathbb{Q}_{(m,k)}$ will converge almost-surely
in the proposed learning algorithm \ref{alg:learning}, i.e.,
$\lim_{t\rightarrow\infty}\mathbb{Q}_{(m,k)}^t=\mathbb{Q}_{(m,k)}^{\infty},\forall
k,m$, where the steady state $\mathbb{Q}$-factor
$\{\mathbb{Q}_{(m,k)}^{\infty}\}$ satisfy:
\begin{equation}
\label{eq:q_infty}\begin{array}{l}
\mathbb{Q}_{(m,k)}^{\infty}(Q_{(m,k)},\mathbf{p})=\\
\quad\quad
T_{(m,k)}^{\mathbb{Q}}(\mathbb{Q}_{(m,k)}^{\infty},Q_{(m,k)},\mathbf{p})-\mathbb{Q}_{(m,k)}^{\infty}(Q_{(m,k)}^I,\mathbf{p}^I_m)
\end{array}
\end{equation}
\begin{proof}
Please refer to Appendix D. 
\end{proof}
\end{Lem}

Note that (\ref{eq:q_infty}) is equivalent to the per-user fixed
point equation for $\mathbb{Q}_{(m,k)}$ in (\ref{eq:Q_linear}). This
result illustrates that the proposed online distributive learning in
(\ref{eq:learn_q_factor}) can converge to the target per user fixed
point solution in (\ref{eq:Q_linear}).

Lemma \ref{lem:converge_value} and \ref{lem:q_converge} only
established the convergence of the proposed online learning
algorithm. Strictly speaking, the converged result is not optimal
due to the linear approximation of the value function
$\widetilde{V}(\widetilde{\mathbf{Q}})$ and the $\mathbb{Q}$-factor
$\mathbb{Q}(\mathbf{Q},\mathbf{p})$ in (\ref{eq:linear_value}) and
(\ref{eq:q_approx}) respectively. The linear approximation is needed
for distributive implementation. As illustrated in Fig.
\ref{fig:opt}, the proposed distributive solution has
close-to-optimal performance compared with brute-force centralized
solution of the Bellman equation in (\ref{eq:bellman_cond}).

\begin{figure}
 \begin{center}
 \resizebox{9cm}{!}{\includegraphics{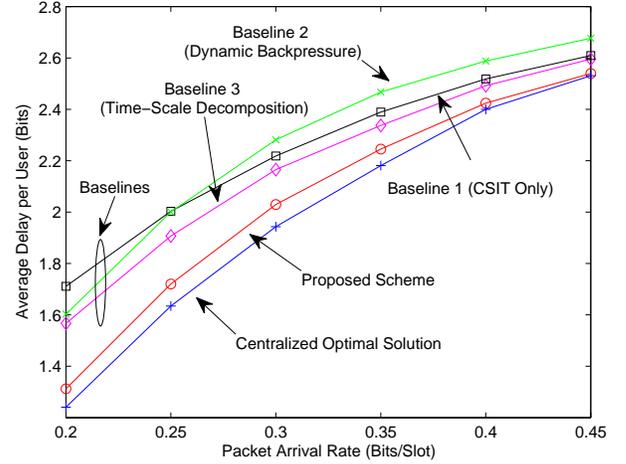}}
 \end{center}
    \caption{Average delay per user versus per user loading $\lambda_{(m,k)}$ in the
Example 1 with the source arrival model is given by
$\Pr\{A_{(m,k)}=1\}=\lambda_{(m,k)}$ and
$\Pr\{A_{(m,k)}=0\}=1-\lambda_{(m,k)}$ for all $m,k$, and the buffer
size $N_Q=3$. Centralized optimal solution refers to the brute-force
centralized solution of the Bellman equation in
(\ref{eq:bellman_cond}). Baseline 1 refers to the {\em CSIT only}
scheme, where the user scheduling are adaptive to the CSIT only.
Baseline 2 refers to the {\em Dynamic Backpressure} scheme
\cite{Georgiadis-Neely-Tassiulas:2006}. Baseline 3 refers to the
{\em time-scale decomposition} scheme proposed in
\cite{pattern:2009}. }
    \label{fig:opt}
\end{figure}

\section{Simulation and Discussion}
\label{sec:simulation}

In this section, we shall compare the proposed distributive
queue-aware intra-cell user scheduling and ICI management control
scheme with three baselines. Baseline 1 refers to the {\em CSIT
only} scheme, where the user scheduling are adaptive to the CSIT
only so as to optimize the achievable data rate. Baseline 2 refers
to a throughput optimal policy (in stability sense) for the user
scheduling, namely the {\em Dynamic Backpressure} scheme
\cite{Georgiadis-Neely-Tassiulas:2006}. In both baseline 1 and 2,
the traditional frequency reuse scheme (frequency reuse factor
equals 3) is used for inter-cell interference management. Baseline 3
refers to the {\em time-scale decomposition} scheme proposed in
\cite{pattern:2009}, where the sets of possible ICI management
patterns $\mathcal{P}$ is the same as the proposed scheme. In the
simulation, we consider a two-tier celluar network composed of 19
BSs as in \cite{pattern:2009}, each has a coverage of 500m. Channel
models are implemented according to the Urban Macrocell Model in
3GPP and Jakes' Rayleigh fading model. Specifically, the path loss
model is given by $PL = 34.5 + 35\log_{10}(r)$, where $r$ (in m) is
the distance from the transmitter to the receiver. The total BW is
10MHz. We consider Poisson packet arrival with average arrival rate
$\mathbb{E}[A_{(m,k)}]=\lambda_{(m,k)}$ (packets/slot) and
exponentially distributed random packet size $\overline{N}_{(m,k)}$
with $\mathbb{E}[\overline{N}_{(m,k)}]=5$Mbits. The scheduling slot
duration $\tau$ is 5ms. The maximum buffer size $N_Q$ is 9 (in
packets), where each user's QSI is partitioned into 4 regions, given
by $\big\{\{0,1,2\};\{3,4,5\};\{6,7\};\{8,9\}\big\}$. The cost
function is given by
$f(Q_{(m,k)})=\frac{Q_{(m,k)}}{\lambda_{(m,k)}}$ for all the users
in the simulations.

\subsection{Performance w.r.t. Transmit Power}
Fig.\ref{fig:delay_tx_pwr} and Fig.\ref{fig:delay_full_buffer}
illustrate the performance of average delay and packet dropping
probability (conditioned on packet arrival) per user versus transmit
power $P_{\max}^m$ respectively. The number of users per BS $K=3$,
and the average arrival rate $\lambda_{(m,k)}=1$. Note that the
average delay and packet dropping probability of all the schemes
decreases as the transmit power increases, and there is significant
performance gain of the proposed scheme compared to all baselines.
This gain is contributed by the QSI-aware user scheduling as well as
ICI management control.

\begin{figure}
 \begin{center}
  \resizebox{9cm}{!}{\includegraphics{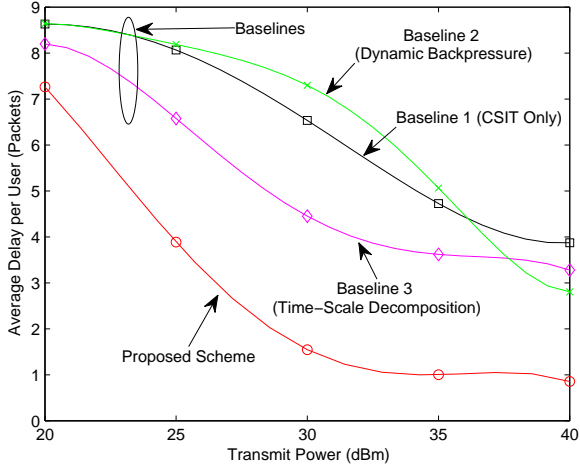}}
 \end{center}
    \caption{Average delay per user versus transmit power $P_{\max}^m$.
    The number of users per BS is $K=3$. The average arrival rate $\lambda_{(m,k)}=1$ (packets/slot). The maximum buffer size $N_Q$ is 9, where each user's QSI is
partitioned into 4 regions, given by
$\big\{\{0,1,2\};\{3,4,5\};\{6,7\};\{8,9\}\big\}$.}
    \label{fig:delay_tx_pwr}
\end{figure}

\begin{figure}
 \begin{center}
  \resizebox{9cm}{!}{\includegraphics{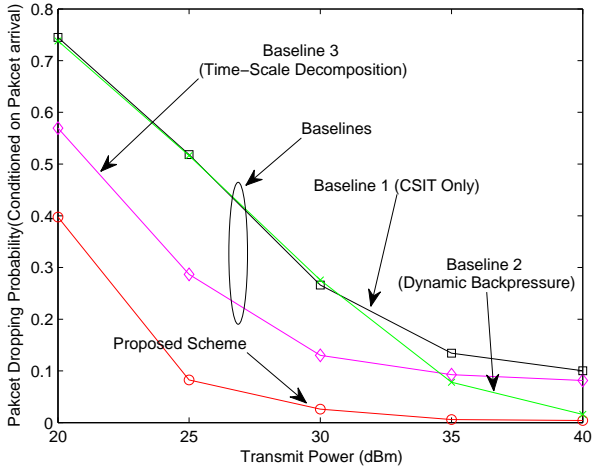}}
 \end{center}
    \caption{Packet dropping probability (conditioned on packet arrival) per user versus transmit power $P_{\max}^m$.
    The number of users per BS is $K=3$. The average arrival rate $\lambda_{(m,k)}=1$ (packets/slot). The maximum buffer size $N_Q$ is 9, where each user's QSI is
partitioned into 4 regions, given by
$\big\{\{0,1,2\};\{3,4,5\};\{6,7\};\{8,9\}\big\}$.}
    \label{fig:delay_full_buffer}
\end{figure}

\subsection{Performance w.r.t. Loading}
Fig.\ref{fig:delay_loading} illustrates the average delay versus per
user loading (average arrival rate $\lambda_{(m,k)}$) at transmit
power of $P_{\max}^m=30$dBm and the number of users per BS $K=3$. It
can also be observed that the proposed scheme achieved significant
gain over all the baselines across a wide range of input loading.

\begin{figure}
 \begin{center}
  \resizebox{9cm}{!}{\includegraphics{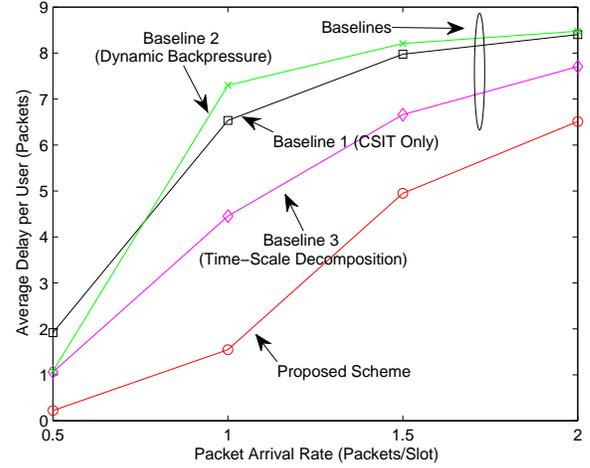}}
 \end{center}
    \caption{Average delay per user versus per user loading
    $\lambda_{(m,k)}$. The transmit power $P_{\max}^m=30$dBm.
    The number of users per BS is $K=3$. The maximum buffer size $N_Q$ is 9, where each user's
QSI is partitioned into 4 regions, given by
$\big\{\{0,1,2\};\{3,4,5\};\{6,7\};\{8,9\}\big\}$.}
    \label{fig:delay_loading}
\end{figure}

\subsection{Cumulative Distribution Function (CDF) of the Queue Length}
Fig.\ref{fig:delay_cdf} illustrates the Cumulative Distribution
Function (CDF) of the queue length per user with transmit power
$P_{\max}^m=25$dBm. The number of users per BS is $K=3$ and the
average arrival rate $\lambda_{(m,k)}=1$. It can be also be verified
that the proposed scheme achieves not only a smaller average delay
but also a smaller delay percentile compared with the other
baselines.

\begin{figure}
 \begin{center}
  \resizebox{9cm}{!}{\includegraphics{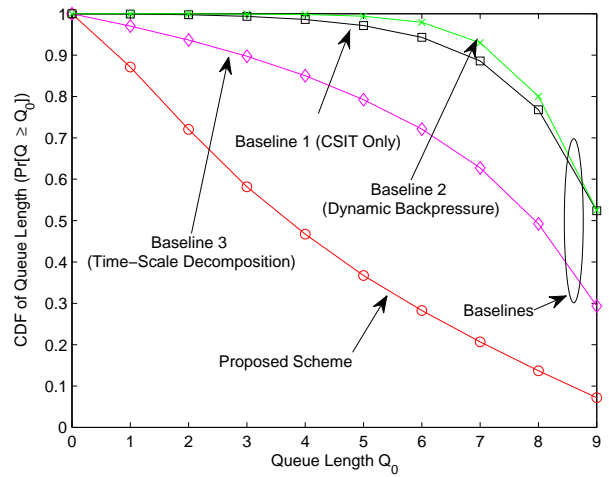}}
 \end{center}
    \caption{Cumulative Distribution Function (CDF) of the queue length per user with transmit power $P_{\max}^m=25$dBm.
    The number of users per BS is $K=3$. The
average arrival rate $\lambda_{(m,k)}=1$. The maximum buffer size
$N_Q$ is 9, where each user's QSI is partitioned into 4 regions,
given by $\big\{\{0,1,2\};\{3,4,5\};\{6,7\};\{8,9\}\big\}$. }
    \label{fig:delay_cdf}
\end{figure}

\subsection{Convergence Performance}
Fig.\ref{fig:delay_learning_vq} illustrates the average delay per
user versus the scheduling slot index with transmit power
$P_{\max}^m=35$dBm. The number of users per BS is $K=3$ and the
average arrival rate $\lambda_{(m,k)}=1.5$. It can be observed that
the convergence rate of the online algorithm is quite fast. For
example, the delay performance of the proposed scheme already
out-performs all the baselines at the $400$-th slot. Furthermore,
the delay performance at $400$-th slot is already quite close to the
converged average delay. Finally, unlike the conventional iterative
NUM approach where the iterations are done offline within the
coherence time of the CSI, the proposed iterative algorithm is
updated over the same time scale of the CSI and QSI updates.
Moreover, the iterative algorithm is online, meaning that useful
payload are transmitted during the iterations.

\begin{figure}
 \begin{center}
  \resizebox{9.5cm}{!}{\includegraphics{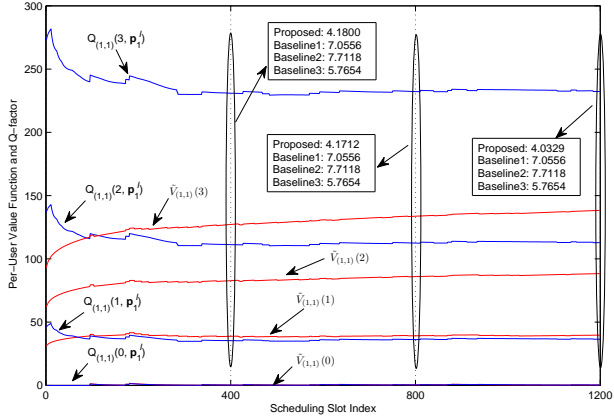}}
 \end{center}
    \caption{Convergence property of the proposed distributive stochastic learning algorithm via stochastic learning. The
transmit power $P_{\max}^m=35$dBm. The number of users per BS is
$K=3$. The average arrival rate $\lambda_{(m,k)}=1.5$. The maximum
buffer size $N_Q$ is 9, where each user's QSI is partitioned into 4
regions, given by $\big\{\{0,1,2\};\{3,4,5\};\{6,7\};\{8,9\}\big\}$.
The figure illustrates instantaneous per-user value function
$\widetilde{V}_{(1,1)}(\widetilde{Q}_{(1,1)})$ and
$\mathbb{Q}$-factor $\mathbb{Q}_{(1,1)}(Q_{(1,1)},\mathbf{p}_1^I)$
versus instantaneous slot index. The boxes indicated the average
delay of various schemes at three selected slot indices.}
    \label{fig:delay_learning_vq}
\end{figure}

\section{Summary}\label{sec:summary}
In this paper, we study the design of a distributive queue-aware
intra-cell user scheduling and inter-cell interference management
control design for a delay-optimal celluar downlink system. We first
model the problem as an infinite horizon average reward POMDP, which
is NP-hard in general. By exploiting special problem structure, we
derive an equivalent Bellman equation to solve the POMDP problem. To
address the distributive requirement and the issue of dimensionality
and computation complexity, we derive a distributive online
stochastic learning algorithm, which only requires local QSI and
local CSI at each of the $M$ BSs. We show that the proposed learning
algorithm converges almost-surely and has significant gain compared
with various baselines. The proposed algorithm only has linear
complexity order $O(MK)$.

\appendices
\section*{Appendix A: Proof of Theorem~\ref{Thm:MDP_cond}}
\label{app:V_increase} Based on the action partitioning, we can
associate the MDP formulation in our delay-optimal control problem
as follows:
\begin{itemize}
\item{\bf State Space:} The system state of the MDP is global QSI
$\mathbf{Q}\in\mathcal{Q}$.
\item{\bf Action Space:} The action on the system state $\mathbf{Q}$ is
the partitioned action $\Omega(\mathbf{Q})$ given in Definition
\ref{def:partitioned action}, and the action space is
$\{\mathcal{P},\mathcal{S}\}$.
\item{\bf Transition Kernel:} The transition kernel
is
$\Pr\{\mathbf{Q}^{\prime}|\mathbf{Q},\Omega(\mathbf{Q})\}=\mathbb{E}_{\mathbf{H}}\left[\Pr\{\mathbf{Q}^{\prime}|\mathbf{Q},\mathbf{H},\Omega(\boldsymbol{\chi})\}\right]$,
where
$\Pr\{\mathbf{Q}^{\prime}|\mathbf{Q},\mathbf{H},\Omega(\boldsymbol{\chi})\}$
is given by (\ref{eq:sys_tran}).
\item{\bf Per-Slot Cost:} The per-slot cost function is
$\hat{g}(\mathbf{Q},\Omega(\mathbf{Q}))=\mathbb{E}_{\mathbf{H}}[g(\mathbf{Q},\mathbf{H},\Omega(\boldsymbol{\chi}))]=\sum_{m,k}\beta_{(m,k)}f(Q_{(m,k)})$.
\end{itemize}

Therefore, the optimal partitioned action $\Omega^*(\mathbf{Q})$ can
be determined from the equivalent Bellman equation in
(\ref{eq:bellman_cond}).

Next, we shall prove that $V(\mathbf{Q})$ is a monotonic increasing
function w.r.t. its component. Given the $V^l(\mathbf{Q})$ is the
result of $l$-th iteration, $V^{l+1}(\mathbf{Q})$ is given by:
\begin{equation}
V^{l+1}(\mathbf{Q})=T_{\Omega}(\mathbf{V}^{l},\mathbf{Q})-T_{\Omega}(\mathbf{V}^{l},\mathbf{Q}^I)
\end{equation}
where
$T_{\Omega}(\mathbf{V}^{l},\mathbf{Q})=\min\limits_{\Omega(\mathbf{Q})}\big[\widetilde{g}(\mathbf{Q},\Omega(\mathbf{Q}))+
\sum_{\mathbf{Q}^{\prime}}\Pr\{\mathbf{Q}^{\prime}|\mathbf{Q},\Omega(\mathbf{Q})\}V^l(\mathbf{Q}^{\prime})\big]$,
and $\mathbf{Q}^I$ is a reference state. Because
$\lim_{l\rightarrow\infty}V^l(\mathbf{Q})=V(\mathbf{Q})$
\cite{Bertsekas:2007}, it is sufficient to prove
$V^l(\mathbf{Q}),\forall l$ is component-wise monotonic increasing.
Using the induction method, we start from $V^0(\mathbf{Q})=0,
\forall \mathbf{Q}$. In the induction step, we assume that $\forall
\mathbf{Q}^1\succ\mathbf{Q}^2,V^l(\mathbf{Q}^1)>V^l(\mathbf{Q}^2)$,
we get
\begin{equation}\begin{array}{ll}
&V^{l+1}(\mathbf{Q}^1)+T_{\Omega}(\mathbf{V}^{l},\mathbf{Q}^I)\\
=&\min\limits_{\Omega(\mathbf{Q}^1)}\left[\widetilde{g}(\mathbf{Q}^1,\Omega(\mathbf{Q}^1))+
\sum\limits_{\mathbf{Q}^{\prime}}\Pr\{\mathbf{Q}^{\prime}|\mathbf{Q}^1,\Omega(\mathbf{Q}^1)\}V^l(\mathbf{Q}^{\prime})\right]\\
>&\sum\limits_{m,k}\beta_{(m,k)}f(Q_{(m,k)}^2)+\sum\limits_{\mathbf{A}}\Pr\{\mathbf{A}\}\mathbb{E}_{\mathbf{H}}
[V(\mathbf{Q}^2-\mathbf{U}^*+\mathbf{A})]\\
\geq&\min\limits_{\Omega(\mathbf{Q}^2)}\left[\widetilde{g}(\mathbf{Q}^2,\Omega(\mathbf{Q}^2))+
\sum\limits_{\mathbf{Q}^{\prime}}\Pr\{\mathbf{Q}^{\prime}|\mathbf{Q}^2,\Omega(\mathbf{Q}^2)\}V^l(\mathbf{Q}^{\prime})\right]\\
=&V^{l+1}(\mathbf{Q}^2)+T_{\Omega}(\mathbf{V}^{l},\mathbf{Q}^I)
\end{array}
\end{equation}
where $\mathbf{U}^*$ is the delivered bits under the conditional
action $\Omega^*(\mathbf{Q}^1)=\{\mathbf{p}^*,\mathbf{s}^*\}$ for
all users. Specifically, $U_{(m,k)}(t)=
R_{(m,k)}(\mathbf{H},\mathbf{p}^*,\mathbf{s}^*)\tau$.

\section*{Appendix B: Proof of Corollary~\ref{cor:dstr_policy}}\label{app:dstr_policy}
Using the linear approximation in (\ref{eq:linear_value}), and the
given ICI management pattern $\mathbf{p}$, the optimal user
scheduling action $\mathbf{s}$ (obtained by solving the RHS of
Bellman equation (\ref{eq:bellman_post})) is:
\begin{equation}\begin{array}{ll}
&\min_{\mathbf{s}(\mathbf{Q},\mathbf{H})\in\mathcal{S}}\Big[\widetilde{g}(\mathbf{Q},\mathbf{P},\mathbf{s}(\mathbf{Q},\mathbf{H}))+\\
&\quad\quad\sum_{\widetilde{\mathbf{Q}}^{\prime}}\Pr\{\widetilde{\mathbf{Q}}^{\prime}|\mathbf{Q},\mathbf{p},\mathbf{s}(\mathbf{Q},\mathbf{H})\}\widetilde{V}(\widetilde{\mathbf{Q}}^{\prime})
\Big]\\
\Rightarrow&\min_{\mathbf{s}(\mathbf{Q},\mathbf{H})\in\mathcal{S}}\bigg[
\sum_{m,k}\Big(\widetilde{V}_{(m,k)}(Q_{(m,k)})(1-s_{(m,k)})+\\
&\quad\quad\widetilde{V}_{(m,k)}((Q_{(m,k)}-U_{(m,k)})^+)s_{(m,k)}\Big)\bigg]\\
\Rightarrow&\max_{\mathbf{s}_m\in\mathcal{S}_m}\sum_{k\in\mathcal{K}_m}\Big(\widetilde{V}_{(m,k)}(Q_{(m,k)})-\\
&\quad\quad\widetilde{V}_{(m,k)}((Q_{(m,k)}-U_{(m,k)})^+)\Big)s_{(m,k)},
\forall m\in\mathcal{M}_{\mathbf{p}}
\end{array}
\end{equation}
where
$\mathcal{S}_m=\{\mathbf{s}_m:\sum_{k\in\mathcal{K}_m}s_{(m,k)}=1,s_{(m,k)}\in\{0,1\}\}$
is the set of all the possible user scheduling policy for BS $m$. As
a result, Corollary \ref{cor:dstr_policy} is obvious from the above
equation.

\section*{Appendix C: Proof of
Lemma~\ref{lem:converge_value}}\label{app:value_converge} From the
definition of mapping $T_{(m,k)}$ in (\ref{eq:T_{(m,k)}}), the
convergence property of the per-user value function update algorithm
in (\ref{eq:learn_value_f}) is equivalent to the following update
equation\cite{Borkar:2008}:
\begin{equation}
\label{eq:syn_learning}\begin{array}{l}
\widetilde{V}_{(m,k)}^{t+1}(\widetilde{Q}_{(m,k)})=\\
\quad\widetilde{V}_{(m,k)}^t(\widetilde{Q}_{(m,k)})+
\gamma(t)\Bigl[T_{(m,k)}(\widetilde{V}_{(m,k)}^{t},\widetilde{Q}_{(m,k)})-\\
\quad\widetilde{V}_{(m,k)}^{t}(\widetilde{Q}_{(m,k)}^I)
-\widetilde{V}_{(m,k)}^{t}(\widetilde{Q}_{(m,k)})+
Z^{t+1}_{(m,k)}(\widetilde{Q}_{(m,k)})\Bigr]
\end{array}
\end{equation}
where
$\widetilde{Z}^{t+1}_{(m,k)}(\widetilde{Q}_{(m,k)})=\beta_{(m,k)}f(\widetilde{Q}_{(m,k)}+A_{(m,k)})
+\widetilde{V}_{(m,k)}^{t}(Q_{(m,k)}^{\prime})-
T_{(m,k)}(\widetilde{V}_{(m,k)}^{t},\widetilde{Q}_{(m,k)})$, and
$Q_{(m,k)}^{\prime}=\widetilde{Q}_{(m,k)}+A_{(m,k)}-U_{(m,k)}$.
$U_{(m,k)}$ is determined by the ICI management control pattern
$\mathbf{\widetilde{p}}_m^I$ and local CSI $\mathbf{H}_{(m,k)}$. Let
$\mathbb{F}_t=\sigma(\mathbf{\widetilde{V}}_{(m,k)}^l,\mathbf{\widetilde{Z}}_{(m,k)}^l,l\leq
t)$ be the $\sigma$-algebra generated by
$\{\mathbf{\widetilde{V}}_{(m,k)}^l,\mathbf{\widetilde{Z}}_{(m,k)}^l,l\leq
t\}$, It can be verified that
$\mathbb{E}_{\{\mathbf{H}_{(m,k)},A_{(m,k)}\}}[
\mathbf{\widetilde{Z}}^{t+1}_{(m,k)}|\mathbb{F}_t]=0$, and
$\mathbb{E}_{\{\mathbf{H}_{(m,k)},A_{(m,k)}\}}[
||\mathbf{\widetilde{Z}}^{t+1}_{(m,k)}||^2|\mathbb{F}_t]\leq
C_1(1+||\widetilde{\mathbf{V}}_{(m,k)}^t||^2)$ for a suitable
constant $C_1$. Therefore, the learning algorithm in
(\ref{eq:syn_learning}) is a standard stochastic learning algorithm
with the Martingale difference noise
$\mathbf{\widetilde{Z}}^{t+1}_{(m,k)}$. We use the ordinary
differential equation (ODE) to analyze the convergence probability.
Specifically, the limiting ODE associated for
(\ref{eq:syn_learning}) to track asymptotically is given by:
\begin{equation}
\label{eq:ode_syn}\begin{array}{l}
\dot{\widetilde{\mathbf{V}}}_{(m,k)}(t)=\mathbf{T}_{(m,k)}(\widetilde{\mathbf{V}}_{(m,k)}(t))-\widetilde{\mathbf{V}}_{(m,k)}(t)-\\
\quad\quad\quad\quad\widetilde{V}_{(m,k)}(\widetilde{Q}_{(m,k)}^I,t)\mathbf{e}=h(\widetilde{\mathbf{V}}_{(m,k)}(t))
\end{array}
\end{equation}
Note that there is a unique fixed point
$\widetilde{\mathbf{V}}^*_{(m,k)}$ that satisfies the Bellman
equation
\begin{equation}
\mathbf{T}_{(m,k)}(\widetilde{\mathbf{V}}^*_{(m,k)})-\widetilde{\mathbf{V}}^*_{(m,k)}-\widetilde{V}_{(m,k)}^{*}(\widetilde{Q}_{(m,k)}^I)\mathbf{e}=0
\end{equation}
and it is proved in \cite{SA:ODE:globally} that
$\widetilde{\mathbf{V}}^*_{(m,k)}$ is the globally asymptotically
stable equilibrium for (\ref{eq:ode_syn}). Furthermore, define
$h_r(\widetilde{\mathbf{V}}_{(m,k)})=h(r\widetilde{\mathbf{V}}_{(m,k)})/r,
\forall r>0$ and
$h_{\infty}(\widetilde{\mathbf{V}}_{(m,k)})=\lim_{r\rightarrow\infty}h_r(\widetilde{\mathbf{V}}_{(m,k)})=\mathbf{P}_{(m,k)}\widetilde{\mathbf{V}}_{(m,k)}-\widetilde{\mathbf{V}}_{(m,k)}-\widetilde{V}_{(m,k)}(\widetilde{Q}_{(m,k)}^I)\mathbf{e}$.
The origin is the globally asymptotically stable equilibrium point
of the ODE
$\dot{\widetilde{\mathbf{V}}}_{(m,k)}(t)=h_{\infty}(\widetilde{\mathbf{V}}_{(m,k)}(t))$
(This is merely a special case by setting
$\widetilde{\mathbf{g}}_{(m,k)}=\mathbf{0}$ in the
$\mathbf{T}_{(m,k)}(\widetilde{\mathbf{V}}_{(m,k)})$). By theorem
2.2 of \cite{SA:ODE:Bound}, the iterates
$\widetilde{\mathbf{V}}^t_{(m,k)}$ remains bounded almost-surely. By
the ODE approach\cite[Chap.2]{Borkar:2008}, we can conclude that the
iterates of the update
$\widetilde{\mathbf{V}}_{(m,k)}^t\rightarrow\widetilde{\mathbf{V}}_{(m,k)}^*$
almost-surely, i.e., converging to the globally asymptotically
stable equilibrium of the associated ODE.

Finally the proof of
$\widetilde{V}_{(m,k)}^{\infty}(\widetilde{Q}_{(m,k)})=\widetilde{V}_{(m,k)}^{*}(\widetilde{Q}_{(m,k)})$
being a monotonic increasing function can be derived in the same way
as Theorem \ref{Thm:MDP_cond}.

\section*{Appendix D: Proof of
Lemma~\ref{lem:q_converge}}\label{app:q_converge} From the
definition of mapping
$T_{(m,k)}^{\mathbb{Q}}(\mathbb{Q}_{(m,k)},Q_{(m,k)},\mathbf{p})$ in
(\ref{eq:Q_map}), defining the vector form mapping
$\mathbf{T}_{(m,k)}^{\mathbb{Q}}(\mathbb{Q}_{(m,k)}):\mathbb{R}^{(1+N_Q)\times|\mathcal{P}|}\rightarrow
\mathbb{R}^{(1+N_Q)\times|\mathcal{P}|}$ where each elements is
given by
$T_{(m,k)}^{\mathbb{Q}}(\mathbb{Q}_{(m,k)},Q_{(m,k)},\mathbf{p})$.
The convergence property of the per-user $\mathbb{Q}$-factor update
algorithm in (\ref{eq:learn_q_factor}) is equivalent to the
following update equation\cite{Borkar:2008}:
\begin{eqnarray}
\label{eq:Q_sync}\begin{array}{l} \mathbb{Q}_{(m,k)}^{t+1}=
\mathbb{Q}_{(m,k)}^{t}+
\gamma(t)\Bigl[\mathbf{T}_{(m,k)}^{\mathbb{Q}}(\mathbb{Q}_{(m,k)}^{t})-\\
\quad\quad\mathbb{Q}_{(m,k)}^{t}(Q_{(m,k)}^I,\mathbf{p}_m^I)\mathbf{e}
-\mathbb{Q}_{(m,k)}^{t}+ \mathbf{Z}^{t+1}_{(m,k)}\Bigr]
\end{array}
\end{eqnarray}
where $\mathbf{Z}^{t+1}_{(m,k)}$ is the vector form of
$Z^{t+1}_{(m,k)}(Q_{(m,k)},\mathbf{p})$,
$Z^{t+1}_{(m,k)}(Q_{(m,k)},\mathbf{p})=\beta_{(m,k)}f(Q_{(m,k)})
+\mathbb{Q}_{(m,k)}^{t}(Q_{(m,k)}^{\prime})-
T_{(m,k)}^{\mathbb{Q}}(\mathbb{Q}_{(m,k)}^{t},Q_{(m,k)},\mathbf{p})$,
and $Q_{(m,k)}^{\prime}=Q_{(m,k)}-U_{(m,k)}+A_{(m,k)}$. $U_{(m,k)}$
is determined by the ICI management control pattern $\mathbf{p}$ and
local CSI $\mathbf{H}_{(m,k)}$. Let
$\mathbb{F}_t=\sigma(\mathbb{Q}_{(m,k)}^l,\mathbf{Z}_{(m,k)}^l,l\leq
t)$ be the $\sigma$-algebra generated by
$\{\mathbb{Q}_{(m,k)}^l,\mathbf{Z}_{(m,k)}^l,l\leq t\}$, It can be
verified that $\mathbb{E}_{\{\mathbf{H}_{(m,k)},A_{(m,k)}\}}[
\mathbf{Z}^{t+1}_{(m,k)}|\mathbb{F}_t]=0$, and
$\mathbb{E}_{\{\mathbf{H}_{(m,k)},A_{(m,k)}\}}[
||\mathbf{Z}^{t+1}_{(m,k)}||^2|\mathbb{F}_t]\leq
C_1(1+||\mathbb{Q}_{(m,k)}^t||^2)$ for a suitable constant $C_1$.
Therefore, the learning algorithm in (\ref{eq:Q_sync}) is also a
standard stochastic learning algorithm with the Martingale
difference noise $\mathbf{Z}^{t+1}_{(m,k)}$. The limiting ODE
associated to track asymptotically is given by:
\begin{equation}
\label{eq:ode_syn_Q}\begin{array}{l}
\dot{\mathbb{Q}}_{(m,k)}(t)=\mathbf{T}_{(m,k)}^{
\mathbb{Q}}(\mathbb{Q}_{(m,k)}(t))-\\
\quad\quad\quad\mathbb{Q}_{(m,k)}(t)-\mathbb{Q}_{(m,k)}(Q_{(m,k)}^I,\mathbf{p}_m^I,t)\mathbf{e}
\end{array}
\end{equation}
Furthermore, there is a unique fixed point $\mathbb{Q}_{(m,k)}^*$
satisfying the following equation \cite{SA:Q_learning}:
\begin{equation}
\mathbf{T}_{(m,k)}^{
\mathbb{Q}}(\mathbb{Q}_{(m,k)}^*)-\mathbb{Q}_{(m,k)}^*-\mathbb{Q}_{(m,k)}^*(Q_{(m,k)}^I,\mathbf{p}_m^I)\mathbf{e}=0
\end{equation}
and it is proved in \cite{SA:Q_learning} that $\mathbb{Q}^*_{(m,k)}$
is the globally asymptotically stable equilibrium for
(\ref{eq:ode_syn_Q}). As a result, following the same argument in
the convergence proof of per-user value function in Lemma
\ref{lem:converge_value}, we can conclude that the iterates of the
update $\mathbb{Q}_{(m,k)}^t\rightarrow\mathbb{Q}_{(m,k)}^*$
almost-surely.

\bibliographystyle{IEEEtran}
\bibliography{IEEEabrv,multicell}

\end{document}